\documentclass[prd,aps,amsfonts,eqsecnum,superscriptaddress,nofootinbib,longbibliography,notitlepage]{revtex4-1}

\newsavebox{\foobox}
\newcommand{\slantbox}[2][0]{\mbox{%
        \sbox{\foobox}{#2}%
        \hskip\wd\foobox
        \pdfsave
        \pdfsetmatrix{1 0 #1 1}%
        \llap{\usebox{\foobox}}%
        \pdfrestore
}}
\newcommand\unslant[2][-.25]{\slantbox[#1]{$#2$}}

\newcommand{\mpi}{\text{\unslant[-.18]\pi}}

\usepackage{graphicx}
\usepackage{xcolor}
\usepackage{rotating}
\usepackage{tikz}
\usepackage{amsmath,amssymb,graphics,amsthm,isomath}

\usepackage[colorlinks=true, urlcolor=violet, linkcolor=blue, citecolor=red, hyperindex=true, linktocpage=true]{hyperref}
\usepackage[capitalise,compress]{cleveref}

\newtheorem{thm}{Theorem}
\numberwithin{thm}{section}
\newtheorem{cor}[thm]{Corollary}
\newtheorem{lem}[thm]{Lemma}

\newtheorem{prop}[thm]{Proposition}

\makeatletter
\renewcommand{\p@subsection}{}
\renewcommand{\p@subsubsection}{}
\makeatother

\usepackage{xcolor}
\usepackage{mathtools}

\usepackage{dsfont}

\begin{document}

\title{Non-perturbative dynamics of the operator size distribution in the Sachdev-Ye-Kitaev model}

\author{Andrew Lucas}
\email{andrew.j.lucas@colorado.edu}
\affiliation{Department of Physics and Center for Theory of Quantum Matter, University of Colorado, Boulder CO 80309, USA}

\begin{abstract}
We prove non-perturbative bounds on the time evolution of the probability distribution of operator size in the $q$-local Sachdev-Ye-Kitaev model with $N$ fermions, for any even integer $q>2$ and any positive even integer $N>2q$.  If the couplings in the Hamiltonian are independent and identically distributed Rademacher random variables, the infinite temperature many-body Lyapunov exponent is almost surely finite as $N\rightarrow\infty$.  In the limit $q \rightarrow \infty$, $N\rightarrow \infty$, $q^{6+\delta}/N\rightarrow 0$, the shape of the size distribution of a growing fermion, obtained by leading order perturbation calculations in $1/N$ and $1/q$, is similar to a distribution that locally saturates our constraints.    Our proof is not based on Feynman diagram resummation; instead, we note that the operator size distribution obeys a continuous time quantum walk with bounded transition rates, to which we apply concentration bounds from classical probability theory.
\end{abstract}
\date{\today}

\maketitle
\tableofcontents

\section{Introduction}
\subsection{Quantum gravity}
The Sachdev-Ye-Kitaev (SYK) model \cite{sachdevye,sachdev15,maldacena2016remarks,suh} is conjectured to be a toy model for quantum gravity in two dimensional spacetime.   The nature of the proposed correspondence is holographic \cite{maldacena}:  this model of quantum mechanics with $N\gg 1$ degrees of freedom should be equivalent to a quantum theory in two spacetime dimensions, which contains dynamical (quantum) gravity with coupling constant $G_{\mathrm{N}} \sim 1/N$  \cite{maldacena2016remarks, suh, almheiri, Jensen, Maldacena1606}.   If this correspondence is true, then at least one model of non-perturbative quantum gravity reduces to the physics of a special quantum mechanical model with a finite number of degrees of freedom.   

A particularly elegant feature of these holographic theories of quantum gravity is the ``maximally rapid" exponential growth of certain out-of-time-ordered correlation functions (OTOCs) \cite{shenker13,stanfordbound}.   If $\mathcal{O}_1$ and $\mathcal{O}_2$ represent few-body operators in the holographic quantum system, 
\begin{equation}
\mathrm{tr}\left(\sqrt{\rho} [\mathcal{O}_1(t),\mathcal{O}_2)]^\dagger \sqrt{\rho} [\mathcal{O}_1(t),\mathcal{O}_2)] \right) \sim \frac{1}{N} \mathrm{e}^{\lambda t},  \label{eq:introlyapunov}
\end{equation}
with $\lambda$ the many-body Lyapunov exponent.  One implicitly assumes that $N$ is very large, so that there might be a parametrically long range of times over which (\ref{eq:introlyapunov}) holds.  Here $\rho$ is the density matrix associated to a mixed quantum state, usually taken to be a thermal ensemble at inverse temperature $\beta$, where it was conjectured in \cite{stanfordbound} that \begin{equation}
\lambda \le \frac{2\mpi}{\beta}
\end{equation}
in many quantum systems -- especially those holographically dual to quantum gravity, where this bound is saturated \cite{shenker13}.  Indeed, the discovery of this maximal growth in the SYK model  demonstrated its apparent connections to quantum gravity.

In certain circumstances, these OTOCs can be interpreted as measuring the average size of a growing operator \cite{nahum, tibor, lucas1809}.  At leading order in $1/N$, there are hints that the time evolution of the operator size distribution in the SYK model \cite{stanford1802, alex1811} may reconstruct its holographically dual geometry \cite{lenny1,lenny2}.

 The expectation that $\lambda < \infty$ as $N\rightarrow \infty$ (at least for broad classes of models) goes by the name of the ``fast scrambling conjecture" \cite{sekino}.  While there are certainly counterexamples to the fast scrambling conjecture \cite{lucas1805, lucas1903} (at least at infinite temperature, where $\rho$ is proportional to the identity), they are arguably rather finely tuned.  For many models, especially those related via the holographic correspondence to quantum gravity, the fast scrambling conjecture is expected to hold.  However, it has never been mathematically demonstrated that the fast scrambling conjecture holds  in a genuine model of (any kind of) quantum gravity.  All that is known is that the fast scrambling conjecture holds at leading non-trivial order in a perturbative $1/N$ expansion of OTOCs.  Unfortunately, this $1/N$ expansion is expected to have zero radius of convergence: see e.g. \cite{stanford1903}.

The purpose of this paper is to prove the fast scrambling conjecture in the SYK model at infinite temperature.  We do so by proving non-perturbative constraints on the dynamics of the operator size distribution.   If the holographic correspondence is true, these are non-perturbative constraints on the early time dynamics in a theory of quantum gravity.  We find that a previous perturbative calculations of operator growth at infinite temperature \cite{stanford1802} essentially saturates our bounds, after a suitable rescaling of time.  So we conjecture that the perturbative description of operator growth in the SYK model is robust to all non-perturbative quantum effects.  This suggests that a semiclassical treatment of many-body chaos in holographic quantum gravity can be sensible, and that non-perturbative quantum gravitational effects do not destroy the emergence of (semi)classical space and time. 

While we focus on the SYK model in this paper, we expect that our bounds on operator growth can be generalized to other random quantum systems on regular factor graphs \cite{chen1}.   We emphasize that our approach completely bypasses the ordinary diagrammatic and manifestly perturbative approach to many-body physics, where non-perturbative effects must be obtained by sophisticated and generally non-unique resummation techniques \cite{stanford1903}.    Specifically, in our approach, non-perturbative results are rigorously obtained by classical combinatorial calculations with $\mathrm{O}(\log N)$ ``loops", which can themselves be bounded in a simple way.  In contrast, in an ordinary diagrammatic expansion, non-perturbative effects often become important at $\mathrm{O}(N)$ loops, which is usually prohibitively challenging to reach in analytic calculations.
 
 \subsection{Comparison to the Lieb-Robinson Theorem}

The techniques we use in our proof are quite different from the conventional Lieb-Robinson bounds \cite{liebrobinson,hastings}.  Here, one assumes a tensor product structure for the Hilbert space $\mathcal{H}$ (though see \cite{rezakhani}).  Define \begin{equation}
V := \lbrace 1,\ldots,N \rbrace \label{eq:Vdef}
\end{equation} 
to be the set of all quantum degrees of freedom, and then define 
\begin{equation}
\mathcal{H} := \bigotimes_{i=1}^N \mathcal{H}_i.
\end{equation}
 Let $X,Y\subset V$, and denote with $\mathrm{dist}(X,Y)$ a ``minimal path length" between any two points induced by the Hamiltonian which generates time evolution (we leave details to the references).  Then the Lieb-Robinson theorem states that there are constants $0 <C,\lambda,\mu<\infty$ such that \begin{equation}
\lVert [A_X(t),B_Y] \rVert \le C\exp[\lambda t - \mu \mathrm{dist}(X,Y)].
\end{equation}
However, we emphasize that the $\lambda$ here is not the same as $\lambda$ in (\ref{eq:introlyapunov}) in typical strictly local lattice models, as this bound is not qualitatively tight \cite{chen1}.

There are two reasons why the conventional Lieb-Robinson approaches do not seem relevant to the SYK model.  Firstly, the SYK model  (defined in Section \ref{sec:ensemble}) has no spatial locality on a lattice \, so $\mathrm{dist}(X,Y) = 0\text{ or }1$.  The primary ``control parameter" of the Lieb-Robinson bounds therefore is irrelevant.  Furthermore, it is not possible to show that $\lambda$ remains finite as $N\rightarrow \infty$, using a conventional combinatorial derivation of Lieb-Robinson bounds.  Secondly,  conventional Lieb-Robinson bounds are on the operator norm (maximal singular value) of commutators, which is far stronger than the trace-like norms of (\ref{eq:introlyapunov}). 

In this paper, we develop a technique to bypass bounding conventional operator norms, and bound (\ref{eq:introlyapunov}) directly whenever $\rho$ is the infinite temperature density matrix.  Our approach is based on interpreting Heisenberg operator growth as a many-body continuous-time quantum walk, which naturally implements the constraints of unitary time evolution in a far stronger way.  In particular, unitarity demands that a certain ``operator size distribution" is normalized as a conventional probability distribution.  The inability for conventional Lieb-Robinson bounds to maintain this normalization is, in large part, responsible for their inability to give sharp bounds on OTOCs such as (\ref{eq:introlyapunov}).     

Precursors to our approach include \cite{chen1,chen2}, which have also attempted to simplify the combinatorics of Lieb-Robinson bounds and obtain stronger results.\footnote{These methods share some similarity to Feynman diagram resummation methods in mathematical kinetic theory \cite{erdos}.   We do not know the extent to which the kinetic methods could be used to further improve our own techniques.  Our ``diagrammatic" resummation is far simpler than the one used to understand kinetic theory.}  We expect that the methods developed in this paper will be broadly applicable to many problems in quantum dynamics where it is not necessary to bound the operator norm of a commutator, such as quantum state transfer \cite{lucas2001}. 

\section{Preliminaries}
\subsection{Majorana fermions}
 Let $q,N\in2\mathbb{Z}^+$ be positive even integers, and define the set $V$ as in (\ref{eq:Vdef}).
Define a finite dimensional quantum mechanical Hilbert space \begin{equation}
\mathcal{H} := \left(\mathbb{C}^2\right)^{\otimes \frac{N}{2}} .
\end{equation}
To each of the $N$ elements $i\in V$, we associate a Majorana fermion $\psi_i$:  a Hermitian operator on $\mathcal{H}$ obeying the anticommutation relations \begin{equation} 
\lbrace \psi_i, \psi_j\rbrace := 2\mathbb{I}[i=j].  \label{eq:anticommutator}
\end{equation}
Here $\mathbb{I}[\cdots]$ denotes the indicator function, which returns 1 if its argument is true and 0 if false.

Let $\mathcal{B} = \mathrm{End}(\mathcal{H})$ be the set of quantum mechanical operators.  This is a complex inner product space if we define the inner product \begin{equation}
(A|B) := 2^{-N/2} \mathrm{tr}_{\mathcal{H}}(A^\dagger B).
\end{equation}
We denote $\lVert A\rVert = \sqrt{(A|A)}$: we emphasize that this is \emph{not} the conventional operator norm.   

\begin{prop}\label{prop1}
An orthonormal basis for $\mathcal{B}$ is spanned by $\lbrace |Y):Y\subseteq V\rbrace$, where \begin{equation}
|Y) := \psi_Y := \prod_{i\in Y}\psi_i.
\end{equation}
We define $\psi_\emptyset = 1$ (the identity operator).  Here and below the product over anticommuting fermions is ordered by the smallest $i$ first.  
\end{prop}
\begin{proof}
We must show that \begin{equation}
(Y_1|Y_2) = \mathbb{I}[Y_1=Y_2].
\end{equation}
If $Y_1=Y_2$, this identity follows immediately from (\ref{eq:anticommutator}).   Letting \begin{equation}
A\triangle B  = A \cup B - A\cap B
\end{equation}
 denote the symmetric difference of two sets, we conclude from (\ref{eq:anticommutator}) that \begin{equation}
 \psi_{Y_1}\psi_{Y_2} = \alpha \psi_{Y_1\triangle Y_2}\text{ for }\alpha \in \lbrace \pm 1 \rbrace.
 \end{equation}
   It remains to show that \begin{equation}
   \mathrm{tr}_{\mathcal{H}}(\psi_Y) = 0\text{, if } Y\ne\emptyset.  \label{eq:trpsi0} 
   \end{equation}
   To obtain (\ref{eq:trpsi0}) when $|Y|$ is an odd integer, we choose any $i\notin Y$ (which must exist as $N$ is even).  Since \begin{equation}
   \psi_Y\psi_i  = (-1)^{|Y|} \psi_i \psi_Y \;\;\; (i\notin Y),  \label{eq:psiYpsii}
   \end{equation}
   we conclude that for $|Y|$ odd \begin{equation}
   0 = \mathrm{tr}_{\mathcal{H}}(\psi_i \lbrace \psi_Y,\psi_i\rbrace) = \mathrm{tr}_{\mathcal{H}}(\psi_i  \psi_Y\psi_i + \psi_i^2 \psi_Y) =\mathrm{tr}_{\mathcal{H}}(2 \psi_Y),  
   \end{equation}
   where in the last step we used cyclic trace properties and (\ref{eq:anticommutator}).   If on the other hand $|Y| > 0$ is even, let $k$ be the maximal element of $Y$ and write $Y= Y_0 \cup \lbrace k\rbrace$.   Then we obtain \begin{equation}
0=\mathrm{tr}_{\mathcal{H}}(\lbrace \psi_{Y_0}, \psi_k \rbrace) = 2\mathrm{tr}_{\mathcal{H}}(\psi_{Y_0} \psi_k) = 2\mathrm{tr}_{\mathcal{H}}(\psi_{Y}).
   \end{equation}
  Thus we obtain (\ref{eq:trpsi0}) for any non-empty set $Y$.
\end{proof}

Let the Hamiltonian $H$ be a Hermitian operator.  $H$ generates a one parameter family of automorphisms, called time evolution, on $\mathcal{B}$: defining the Liouvillian \begin{equation}
\mathcal{L} = \mathrm{i}[H,\cdot], \label{eq:Liouvillian}
\end{equation}
for any $\mathcal{O}\in\mathcal{B}$ we define \begin{equation}
|\mathcal{O}(t)) = \mathrm{e}^{\mathcal{L}t} |\mathcal{O}).
\end{equation}
In fact, $\mathrm{e}^{\mathcal{L}t}$ is unitary, since \begin{equation}
(\mathcal{O}(t)|\mathcal{O}(t)) = 2^{-N/2} \mathrm{tr}_{\mathcal{H}}\left( \left(\mathrm{e}^{\mathrm{i}Ht}\mathcal{O}\mathrm{e}^{-\mathrm{i}Ht}\right)^\dagger\left(\mathrm{e}^{\mathrm{i}Ht}\mathcal{O}\mathrm{e}^{-\mathrm{i}Ht}\right)\right) =  2^{-N/2}\mathrm{tr}_{\mathcal{H}}\left( \mathcal{O}^\dagger \mathcal{O}\right) = (\mathcal{O}|\mathcal{O}). \label{eq:Lunitary}
\end{equation}
More generally, using the cyclic properties of the trace, we conclude that for any $A,B\in\mathcal{B}$: \begin{equation}
(A|\mathcal{L}|B) = -(B|\mathcal{L}|A). \label{eq:Lantisymmetric}
\end{equation}

Define the projection matrices\begin{equation}
\mathbb{Q}_s |Y) = \mathbb{I}[|Y| = s] |Y).
\end{equation}
Note that \begin{equation}
\sum_{s=0}^N \mathbb{Q}_s = 1 \label{eq:sumPs}
\end{equation}
(with 1 the identity of $\mathrm{End}(\mathcal{B})$).  We say that the non-null vectors of $\mathbb{Q}_s$ correspond to \emph{operators of size $s$}.   
Given $\mathcal{O}\in \mathcal{B}$, we say that \begin{equation}
P_s(\mathcal{O},t) = \frac{(\mathcal{O}(t)| \mathbb{Q}_s |\mathcal{O}(t))}{(\mathcal{O}(t)|\mathcal{O}(t))} \label{eq:PsOt}
\end{equation}
is the probability that operator $\mathcal{O}$ is size $s$ at time $t$.   To see that this is a well-defined probability measure on $\lbrace 0, 1,\ldots, N\rbrace$, observe that $\mathbb{Q}_s$ is positive semidefinite and hence $P_s \ge 0$; then from  (\ref{eq:sumPs}), \begin{equation}
\sum_{s=0}^N P_s(\mathcal{O},t) = 1,
\end{equation} 
for any $\mathcal{O}\in \mathcal{B}$ and $t\in\mathbb{R}$.    For simplicity, we will drop the explicit $\mathcal{O}$ in $P_s(t)$, as our formalism does not depend on the particular choice of operator.

%

\subsection{The SYK ensemble}\label{sec:ensemble}
The SYK model corresponds to a random ensemble of Hamiltonians.  Define \begin{equation}
F := \lbrace X\subset V : |X| = q\rbrace
\end{equation}
to be the set of all subsets of $V$ with exactly $q$ elements.
For each $X\in F$, let $J_X$ be an independent and identically distributed (iid) Rademacher\footnote{In the physics literature, the random variables $J_X$ are typically taken to be Gaussian.  We expect that a very similar result to ours will hold in this case as well, but we found the combinatorial problem discussed in Section \ref{sec:trans} to be a bit more elegant for Rademacher random variables.} random variable: \begin{equation}
\mathbb{P}\left[J_X = \sigma \right] = \mathbb{P}\left[J_X = -\sigma \right]  = \frac{1}{2}, 
\end{equation}
where \begin{equation}
\sigma := \left[2q\left(\begin{array}{c} N-1 \\ q-1 \end{array}\right) \right]^{-1/2}.  \label{eq:sigmadef}
\end{equation}
The $q$-local SYK model is the random ensemble of Hamiltonians $H$, corresponding to the random Hermitian matrix
\begin{equation}
H := \mathrm{i}^{q/2} \sum_{X\in F} J_X \prod_{i \in X} \psi_i := \sum_{X\in F} H_X.
\end{equation}
The randomness in the SYK ensemble is essential in our proof of operator growth bounds.   Averages over the ensemble of random variables $\lbrace J_X\rbrace $ are denoted as $\mathbb{E}[\cdots]$, and probability is denoted as $\mathbb{P}[\cdots]$, as above.
We define $\mathcal{L}_X := \mathrm{i}[H_X,\cdot]$.
\begin{prop} \label{propk}
If $\mathbb{Q}_s|\mathcal{O}_s) = |\mathcal{O}_s)$, $X\in F$, and $\mathbb{Q}_{s^\prime} \mathcal{L}_X|\mathcal{O}_s) \ne 0$, then there exists $k\in\mathbb{Z}^+$ for which \begin{equation}
s^\prime -s = q+2-4k.  \label{eq:kdef}
\end{equation}  
In particular, \begin{equation}
|s^\prime - s| \le q-2. \label{eq:qm2max}
\end{equation}
\end{prop}
\begin{proof}
Since $\mathcal{L}_X|Y)$ is proportional to $|[\psi_X,\psi_Y])$, we analyze when $[\psi_X,\psi_Y]\ne 0$ is possible.   Without loss of generality we write \begin{equation}
Z = X\cap Y, \;\;\;\;\; V=X-Z, \;\;\;\;\; W = Y-Z,
\end{equation}in which case it suffices to constrain
\begin{equation}
\lVert [\psi_X,\psi_Y]\rVert = \lVert [\psi_V\psi_Z,\psi_W\psi_Z]\rVert  = \lVert [\psi_V\psi_Z,\psi_W]\psi_Z + \psi_W [\psi_V,\psi_Z]\psi_Z\rVert.
\end{equation}
By repeated use of (\ref{eq:psiYpsii}), if $A\cap B = 0$, \begin{equation}
\psi_A \psi_B = (-1)^{|A||B|}\psi_B\psi_A, \label{eq:psiApsiB}
\end{equation}
hence $[\psi_A,\psi_B]\ne 0$ if and only if $|A|$ and $|B|$ are both odd.   Since $|V\cap Z|$ is even, we conclude that $[\psi_X,\psi_Y]=0$ unless $|V|$ and $|Z|$ are both odd, in which case $|[\psi_X,\psi_Y])$ is proportional to $|\psi_V\psi_W)$.  If $X\in F$, then $|X|=q$, and so $\mathbb{Q}_{s^\prime}\mathcal{L}_X\mathbb{Q}_s|Y)\ne 0$ only if $|Y|=s$ and \begin{equation}
s^\prime = |X|+|Y|-2|X\cap Y| = s+q - 2|X\cap Y|.
\end{equation}
Since $|X\cap Y|$ is odd, we obtain the desired result.
\end{proof}
Since by definition $s^\prime > s$, we conclude that \begin{equation}
2k-1 < \frac{q}{2}.  \label{eq:2kminus1}
\end{equation}

It will be useful to define the following partition of the set of all non-trivial operator sizes: \begin{equation}
\lbrace 1,\ldots ,N \rbrace = \bigcup_{l=0}^{N^\prime} R_l \label{eq:ldef}
\end{equation}
where \begin{equation}
N^\prime := \left\lceil \frac{N-1}{q-2} \right\rceil \label{eq:Nprimedef}
\end{equation}
and \begin{equation}
R_l := \left\lbrace \begin{array}{ll}  \lbrace 1\rbrace &\ l=0 \\ \lbrace m \in \mathbb{Z}:  (l-1)(q-2)+1<m\le  l(q-2)+1 \rbrace &\ 0<l<N^\prime \\
 \lbrace m \in \mathbb{Z}:  (N^\prime-1)(q-2)+1<m\le  N \rbrace &\  l = N^\prime  \end{array} \right..
\end{equation}
Analogous to before, we say that the probability that $\mathcal{O}(t)$ is in \emph{block} $l$ is given by \begin{equation}
P_l(t) :=  \sum_{s\in R_l} P_s(t).
\end{equation} 
We define the projectors \begin{equation}
\mathbb{Q}_l := \sum_{s\in R_l} \mathbb{Q}_s
\end{equation}
We use the same symbols $P$ and $\mathbb{Q}$ to denote operator probability distributions and projectors, and will use the subscript ($l$ or nearby letters vs. $s$ or nearby letters) to distinguish whether we are referring to blocks or individual sizes.

\section{Operator growth}
\subsection{Quantum walk of the size distribution}
We now show that the time evolution of the operator block distributions $P_l(t)$ can be interpreted as a continuous time quantum walk on a finite one dimensional chain.   The results of this section are very general, and apply to a broad class of quantum many-body models beyond the SYK model, for suitable redefinitions of operator size.  Define $\varphi_l(t)$ as\footnote{The authors of \cite{stanford1802} called this the ``operator wave function".} \begin{equation}
\varphi_l(t) := \sqrt{P_l(t)}.   \label{eq:varphisqrt}
\end{equation}  For $\mathcal{M} \in \mathrm{End}(\mathcal{B})$, we define \begin{equation}
\lVert \mathcal{M}\rVert := \sup_{\mathcal{O} \in \mathcal{B}} \frac{\lVert \mathcal{M}|\mathcal{O})\rVert}{\lVert \mathcal{O}\rVert} = \sup_{\mathcal{O},\mathcal{O}^\prime \in \mathcal{B}} \frac{| (\mathcal{O}^\prime| \mathcal{M}|\mathcal{O}) |}{\lVert \mathcal{O}\rVert \lVert \mathcal{O}^\prime \rVert} ,  \label{eq:endBnorm}
\end{equation}
in analogy with the usual operator norm on linear transformations. 
  \begin{prop}\label{prop2}
Let $0\le s<s^\prime \le N$, let $H$ be a Hamiltonian drawn from the SYK ensemble, and let \begin{equation}
\mathcal{K}_{s^\prime s} = \lVert \mathbb{Q}_{s^\prime}\mathcal{L}\mathbb{Q}_s\rVert.  \label{eq:calKss}
\end{equation}
Then there exist functions $K_{s^\prime s}: \mathbb{R} \rightarrow [-\mathcal{K}_{ss^\prime}, \mathcal{K}_{ss^\prime} ]$ such that \begin{equation}
\frac{\mathrm{d}}{\mathrm{d}t} \sqrt{P_s(t)}= \sum_{s^\prime < s} K_{s s^\prime}(t) \sqrt{P_{s^\prime}(t)} -  \sum_{s>s^\prime} K_{ s^\prime s}(t) \sqrt{P_{s^\prime}(t)} \label{eq:prop2}
\end{equation}
Analogously, there exist functions $K_l: \mathbb{R} \rightarrow [-\mathcal{K}_l, \mathcal{K}_l]$ such that \begin{equation}
\frac{\mathrm{d}\varphi_l}{\mathrm{d}t} = K_{l-1}(t) \varphi_{l-1}(t) - K_l(t) \varphi_{l+1}(t), \label{eq:prop2dos}
\end{equation}
(recall $l$ was defined in (\ref{eq:ldef})) so long as \begin{equation}
\mathcal{K}_l = \max\left(\max_{s\in R_l} \sum_{s^\prime \in R_{l+1}} \mathcal{K}_{s^\prime s},  \max_{s^\prime \in R_{l+1}} \sum_{s\in R_{l}} \mathcal{K}_{s^\prime s} \right)  \label{eq:Kldef}
\end{equation}
and $K_{-1}(t) = K_{N^\prime}(t)=0$.  These latter restrictions simply restrict the dynamics to operators in blocks $R_0$ to $R_{N^\prime}$.
\end{prop}
\begin{proof}
For simplicity in this proof, the $t$-dependence of $\mathcal{O}$ is implicit; without loss of generality, we may take $\lVert \mathcal{O}\rVert = 1$ by (\ref{eq:Lunitary}).  For $s\in \lbrace0,\ldots,N\rbrace$, let $|\mathcal{A}_s)$ be a unit norm operator such that \begin{equation}
\mathbb{Q}_s|\mathcal{O}) = \sqrt{P_s}|\mathcal{A}_s),
\end{equation}
and note that if $P_s \ne 0$, $|\mathcal{A}_s)$ is unique.   Now from (\ref{eq:Liouvillian}) and (\ref{eq:PsOt}), \begin{align}
\frac{\mathrm{d}P_s}{\mathrm{d}t} &= (\mathcal{O}| [\mathbb{Q}_s,\mathcal{L}] |\mathcal{O}) = \sum_{s^\prime} \sqrt{P_sP_{s^\prime}} \left[(\mathcal{A}_s|\mathcal{L}|\mathcal{A}_{s^\prime})-(\mathcal{A}_{s^\prime}|\mathcal{L}|\mathcal{A}_{s})\right] \notag \\
&= 2\sqrt{P_s} \sum_{s^\prime < s} K_{ss^\prime}(t)\sqrt{P_{s^\prime}} - 2\sqrt{P_s} \sum_{s^\prime > s} K_{s^\prime s}(t)\sqrt{P_{s^\prime}} \label{eq:prop2first}
\end{align}
where in the first line, we used (\ref{eq:sumPs}) and in the second line we used (\ref{eq:Lantisymmetric}) and defined \begin{equation}
K_{ss^\prime}(t) = (\mathcal{A}_s|\mathbb{Q}_s\mathcal{L}\mathbb{Q}_{s^\prime}|\mathcal{A}_{s^\prime}). \label{eq:proofKss}
\end{equation}
Since $\mathrm{d}\sqrt{P_s} = \mathrm{d}P_s / 2\sqrt{P_s}$, we obtain (\ref{eq:prop2}).  Combining (\ref{eq:endBnorm}), (\ref{eq:calKss}) and (\ref{eq:proofKss}), we obtain (\ref{eq:prop2}). 

The analogue result for block probabilities is identically derived.  In addition, observe that (\ref{eq:qm2max}) implies that \begin{equation}
\mathbb{Q}_{l^\prime} \mathcal{L} \mathbb{Q}_l \ne 0 \text{ only if } |l^\prime - l| \le 1.
\end{equation}
Hence we obtain (\ref{eq:prop2dos}) where \begin{equation}
K_l(t) := (\mathcal{O}(t) | \mathbb{Q}_{l+1} \mathcal{L}\mathbb{Q}_l | \mathcal{O}(t)) \le \lVert \mathbb{Q}_{l+1}\mathcal{L}\mathbb{Q}_l \rVert := \mathcal{K}_l.
\end{equation}
Using (\ref{eq:endBnorm}): 
\begin{align}
\mathcal{K}_l &= \sup_{\mathcal{O},\mathcal{O}^\prime} \frac{(\mathcal{O}^\prime|\mathbb{Q}_{l+1}\mathcal{L}\mathbb{Q}_l|\mathcal{O})}{\sqrt{(\mathcal{O}|\mathcal{O})(\mathcal{O}^\prime|\mathcal{O}^\prime)}} \le \sup_{\mathcal{O},\mathcal{O}^\prime} \sum_{\substack{ s\in R_l \\s^\prime \in R_{l+1} } }\sqrt{P_s(\mathcal{O})P_{s^\prime}(\mathcal{O}^\prime)} \lVert \mathbb{Q}_{s^\prime} \mathcal{L}\mathbb{Q}_s\rVert \notag \\
&\le \sup_{\mathcal{O},\mathcal{O}^\prime} \sum_{\substack{ s\in R_l \\s^\prime \in R_{l+1} } }\frac{P_s(\mathcal{O})+P_{s^\prime}(\mathcal{O}^\prime)}{2} \mathcal{K}_{s^\prime s}.
\end{align}
A simple identity leads to (\ref{eq:Kldef}).
 \end{proof}

(\ref{eq:prop2dos}) can be interpreted as follows.  $\varphi_l(t)$ are the coefficients of the real-valued quantum wave function $|\varphi(t)\rangle$ of an auxiliary quantum mechanical system defined on the Hilbert space \begin{equation}
\mathcal{H}_{\mathrm{aux}} := \mathbb{C}^{1+N^\prime} :=  \mathrm{span}\lbrace |0\rangle, |1\rangle, \ldots, |N^\prime\rangle \rbrace;
\end{equation} 
the latter basis states are defined such that \begin{equation}
\varphi_l(t) := \langle l|\varphi(t)\rangle.  \label{eq:auxvarphil}
\end{equation}
The auxiliary Hamiltonian is \begin{equation}
H_{\mathrm{aux}}(t) := \sum_{l=0}^{N^\prime - 1} \mathrm{i} K_l(t) \left( |l\rangle\langle l-1| - |l-1\rangle\langle l| \right). \label{eq:Haux}
\end{equation}
The Schr\"odinger equation for this auxiliary quantum system is (\ref{eq:prop2dos}).   

\subsection{Lyapunov exponent}
Define the operator (block) size distribution \begin{equation}
\mathbb{E}_{\mathrm{s},t} \left[ f(l) \right] := \sum_{l=0}^{N^\prime} f(l) P_l(t).
\end{equation}   A formal definition of the many-body Lyapunov exponent, heuristically defined in (\ref{eq:introlyapunov}), is given by the growth rate of the logarithm of the average operator size $\mathbb{E}_{\mathrm{s},t}[l]$ (recall $l$ was defined in (\ref{eq:ldef}).  This Lyapunov growth is constrained by the following theorem, which is our first main result:
\begin{thm} \label{lyapunovtheor}
Suppose that there exist $c\in \mathbb{R}^+$ and $M \in \mathbb{Z}^+$ such that \begin{equation}
\mathcal{K}_l \le c(l+1)\;\;\;\; \text{ if }l \le M. \label{eq:Kll}
\end{equation}
Then for any $\epsilon \in \mathbb{R}^+$, the many-body Lyapunov exponent obeys \begin{equation}
\frac{\log \mathbb{E}_{\mathrm{s},t}[l]}{t} := \lambda(t) \le 2c (1+\epsilon)  \label{eq:lambdadef}
\end{equation}
for times \begin{equation}
|t| < \frac{1}{4c(1+\mathrm{e})}\left[ \log M - 2 - \log \log \frac{N^{\prime 3}}{2\epsilon}\right] .  \label{eq:scramblingtime}
\end{equation}
\end{thm}
\begin{proof}
Without loss of generality we assume $t\ge 0$.  We begin with the following lemma.    Note that here and below, we write $\mathbb{E}_{\mathrm{s},t}$ as $\mathbb{E}_{\mathrm{s}}$ for convenience.
 \begin{lem}
If (\ref{eq:Kll}) holds with $M\ge N^\prime-1$, then for $n\in \mathbb{Z}^+$, 
\begin{equation}
\frac{\mathrm{d}}{\mathrm{d}t} \mathbb{E}_{\mathrm{s}} \left[l^n \right] \le 4cn (1+\mathrm{e}) \left(\mathbb{E}_{\mathrm{s}}\left[l^n \right] + (\mathrm{e}n)^n\right)  \label{eq:ddtln}
\end{equation}
In the special case $n=1$, the following stronger inequality holds: \begin{equation}
\frac{\mathrm{d}}{\mathrm{d}t} \mathbb{E}_{\mathrm{s}} \left[l \right] \le c \left(2\mathbb{E}_{\mathrm{s}} \left[l \right]  + 1\right).  \label{eq:ddtln1}
\end{equation}\label{lmael}
\end{lem}
\begin{proof}
We begin by using (\ref{eq:prop2dos}):  for any non-decreasing function $f:\mathbb{Z} \rightarrow \mathbb{R}$, \begin{align}
\frac{\mathrm{d}}{\mathrm{d}t} \mathbb{E}_{\mathrm{s}} \left[f(l) \right]  &=  \sum_{l=0}^{N^\prime } f(l) \left(2 \varphi_l \frac{\mathrm{d}\varphi_l}{\mathrm{d}t}\right) = 2\sum_{l=0}^{N^\prime} f(l) \varphi_l \left[ K_{l-1}\varphi_{l-1} - K_l \varphi_l\right] = 2 \sum_{l=0}^{N^\prime- 1} K_l \varphi_l \varphi_{l+1} [ f(l+1)-f(l)] \notag \\
&\le 2c\sum_{l=0}^{N^\prime- 1}  \varphi_l \varphi_{l+1} (l+1) [ f(l+1)-f(l)] \le c\sum_{l=0}^{N^\prime - 1} (P_l + P_{l+1}) (l+1) [f(l+1)-f(l)]. \label{eq:fl1fl}
\end{align}
In particular, choosing $f(l)=l^n$, we may further loosen this inequality using elementary inequalities: \begin{equation}
\frac{\mathrm{d}}{\mathrm{d}t} \mathbb{E}_{\mathrm{s}} \left[f(l) \right]  \le 2c \sum_{l=0}^{N^\prime -1} P_l \left( (l+1)^{n+1} - l^{n+1}\right). \label{eq:ddtln2}
\end{equation}
Now observe that \begin{equation}
(l+1)^{n+1}-l^{n+1} = (n+1) l^n + \sum_{k=0}^{n-1} \left(\begin{array}{c} n+1 \\ k \end{array}\right) l^k \le (n+1)l^n + n(n+1) (l+1)^{n-1}.
\end{equation}
Next, note the inequality \begin{equation}
n(l+1)^{n-1} < \mathrm{e} l^n +  (\mathrm{e}n)^n \label{eq:lma5ineqinterm}
\end{equation}
which we derive by multiplying both sides of (\ref{eq:lma5ineqinterm}) by $l^{-n}$, assuming $l>1$ (the inequality is trivial when $l=0$): \begin{equation}
\frac{n}{l} \left(1+\frac{1}{l}\right)^{n-1} < \frac{n}{l} \mathrm{e}^{n/l} < \mathrm{e} +  \left(\frac{\mathrm{e}n}{l}\right)^n.
\end{equation}
For $n\le l$, the first term on the right hand side is always at least as large as the middle term; for $n> l$, the second term on the right is larger.  Combining (\ref{eq:ddtln2}) and (\ref{eq:lma5ineqinterm}), we obtain (\ref{eq:ddtln}). 

For the case $n=1$, we use that $f(l+1)-f(l)=1$.  Directly plugging in to (\ref{eq:fl1fl}) we obtain (\ref{eq:ddtln1}).
\end{proof}
The next lemma shows that even when $K_l$ grow faster than (\ref{eq:Kll}) at large $l$, $P_l(t)$ is very small for $l>M$ at early times. 
\begin{lem}\label{lmalarge}
If $K_l(t)$ obeys (\ref{eq:Kll}), then \begin{equation}
\frac{\mathrm{d}}{\mathrm{d}t} \mathbb{P}_{\mathrm{s}}[ l>M] \le 2\mathrm{e}c^2(M+1)t \exp\left[ -M \mathrm{e}^{-2-4c(1+\mathrm{e})t} \right].  \label{eq:PslM}
\end{equation}
\end{lem}
\begin{proof}
We begin by employing (\ref{eq:prop2dos}): \begin{equation}
\frac{\mathrm{d}}{\mathrm{d}t} \mathbb{P}_{\mathrm{s}}[ l>M] = 2\sum_{l=M+1}^{N^\prime} \varphi_l (K_{l-1} \varphi_{l-1} - K_{l+1}\varphi_{l+1}) = 2K_M \varphi_M \varphi_{M+1} \le 2c(M+1)\varphi_{M+1}.  \label{eq:PsgM}
\end{equation}
In the last inequality, we used (\ref{eq:Kll}) along with $\varphi_l(t)\le 1$ for any $l$.  Hence, to obtain (\ref{eq:PslM}), it suffices to bound $\varphi_{M+1}(t)$. 

Let $K \in \mathbb{R}^{(N^\prime+1) \times (N^\prime+1)}$ correspond to the transition matrix whose entries are \begin{equation}
K_{l^\prime l}(t) = K_l \mathbb{I}(l=l^\prime - 1) - K_{l^\prime}(t) \mathbb{I}(l^\prime=l - 1).
\end{equation}
(indices run from $l=0$ to $l=N^\prime$).   Hence $K$ is tridiagonal and antisymmetric.  Let us define the orthogonal matrix $U(t,t^\prime)$  by the differential equation \begin{equation}
\frac{\mathrm{d}}{\mathrm{d}t} U(t,t^\prime)  = K(t) U(t,t^\prime) , \;\;\;\;\;\; U(t^\prime,t^\prime) = 1.
\end{equation}
$U(t,t^\prime)$ generates the continuous time quantum walk with transition rates $K_l(t)$.

Next, we define the quantum walk transition matrix $\widetilde{K}(t)$ as follows: \begin{equation}
\widetilde{K}_{l^\prime l}(t) := K_l \mathbb{I}(M>l=l^\prime - 1) - K_{l^\prime}(t) \mathbb{I}(M>l^\prime=l - 1).
\end{equation}
This matrix corresponds to excising the sites $l>M$ from the walk.  We define an analogous time evolution operator \begin{equation}
\frac{\mathrm{d}}{\mathrm{d}t} \widetilde{U}(t,t^\prime)  = \widetilde{K}(t) \widetilde{U}(t,t^\prime) , \;\;\;\;\;\; \widetilde{U}(t^\prime,t^\prime) = 1.
\end{equation}

Now we use the following integral identity:\footnote{In the physics literature, this is called the integral form of the memory matrix formula \cite{zwanzig, mori, forster}.} \begin{align}
\varphi_{M+1}(t) &= U_{M+1,0}(t,0) = \widetilde{U}_{M+1,0}(t,0) + \int\limits_0^t \mathrm{d}t^\prime \sum_{l,l^\prime}  U_{M+1,l^\prime}(t,t^\prime) (K_{l^\prime l}(t) - \widetilde{K}_{l^\prime l}(t) )  \widetilde{U}_{l,0}(t^\prime,0).  \label{eq:memorymatrix}
\end{align}
Due to the fact that $\widetilde{U}$ does not evolve into sites with $l>M$, we can immediately simplify (\ref{eq:memorymatrix}):
\begin{equation}
\varphi_{M+1}(t) = \int\limits_0^t \mathrm{d}t^\prime \; K_M(t^\prime) \; U_{M+1,M+1}(t,t^\prime)\widetilde{U}_{M,0}(t^\prime,0).
\end{equation}
Using (\ref{eq:Kll}) along with orthogonality of $U(t,t^\prime)$ and the triangle inequality: \begin{equation}
\varphi_{M+1}(t) \le c(M+1) \int\limits_0^t \mathrm{d}t^\prime \; \widetilde{U}_{M,0}(t^\prime,0).  \label{eq:memorymatrix2}
\end{equation}

We now recognize that \begin{equation}
\widetilde{U}_{M,0}(t^\prime,0) =\widetilde{\varphi}_M(t^\prime) \label{eq:widetildevarphi}
\end{equation}
is the solution to the blocked quantum walk generated by $\widetilde{K}$.  This blocked quantum walk obeys Lemma~\ref{lmael}; integrating (\ref{eq:ddtln}), we obtain \begin{equation}
\mathbb{E}_{\widetilde{\mathrm{s}}} \left[ l^n \right] \le (\mathrm{e}n)^n \left( \mathrm{e}^{4c(1+\mathrm{e})nt}-1 \right).
\end{equation}
Here $\mathbb{E}_{\widetilde{\mathrm{s}}} [\cdots]$ denotes averages in the probability distribution of the blocked quantum walk.  Using Markov's inequality, \begin{equation}
\widetilde{\varphi}_M(t) \le \inf_{n\in\mathbb{Z}^+}  \frac{\mathbb{E}_{\widetilde{\mathrm{s}}} \left[ l^n \right]}{M^n} < \inf_{n\in\mathbb{Z}^+} \left(\frac{\mathrm{e}^{1+4c(1+\mathrm{e})t}n}{M} \right)^n \le \exp\left[1 -M \mathrm{e}^{-2-4c(1+\mathrm{e})t} \right], \label{eq:widetildevarphiM}
\end{equation}
where in the last step we used the following sequence of inequalities for $z\in \mathbb{R}^+$:  \begin{equation}
\inf_{n\in\mathbb{Z}^+} \left(\frac{n}{z}\right)^n \le \left(\frac{1}{z} \left\lfloor \frac{z}{\mathrm{e}}\right\rfloor \right)^{\lfloor z/\mathrm{e}\rfloor} \le \exp \left[  - \left\lfloor \frac{z}{\mathrm{e}}\right\rfloor\right] < \exp \left[ 1- \frac{z}{\mathrm{e}}\right] . 
\end{equation}

Combining (\ref{eq:PsgM}), (\ref{eq:memorymatrix2}), (\ref{eq:widetildevarphi}) and (\ref{eq:widetildevarphiM}), and using the fact that our bound on $\widetilde{\varphi}_M(t)$ is a monotonically increasing function of time,  we obtain (\ref{eq:PslM}).
\end{proof}

The last step is to combine (\ref{eq:ddtln1}) with Lemma \ref{lmalarge} to bound the true Lyapunov exponent.   Defining the non-decreasing functions \begin{subequations}\begin{align}
f_>(l) &:= (l-M) \mathbb{I}[l>M], \\
f_<(l) &:= l - f_>(l),
\end{align}\end{subequations} we write \begin{equation}
\mathbb{E}_{\mathrm{s}}[l] = \mathbb{E}_{\mathrm{s}}[ f_<(l) + f_>(l)]
\end{equation}
and bound each piece separately.   Using the fact that (\ref{eq:PslM}) is an increasing function of $t$: \begin{equation}
\mathbb{E}_{\mathrm{s}}[f_>(l)] \le (N^\prime - M) \mathbb{P}_{\mathrm{s}}[l>M] \le 2\mathrm{e}c^2(M+1)N^\prime t^2 \exp\left[ -M \mathrm{e}^{-2-4c(1+\mathrm{e})t} \right].
\end{equation}
Then using (\ref{eq:fl1fl}) and $f_<(l+1) \le f_<(l) + 1$: \begin{equation}
\frac{\mathrm{d}}{\mathrm{d}t} \mathbb{E}_{\mathrm{s}}[f_<(l)] \le c \sum_{l=0}^{N^\prime } (2l+1) P_l(t) = c\left(2\mathbb{E}_{\mathrm{s}} \left[l \right]  + 1\right).
\end{equation}
We conclude that \begin{equation}
\mathbb{E}_{\mathrm{s}}[l] \le \frac{\mathrm{e}^{2ct}-1}{2} + 2\mathrm{e}c^2(M+1)N^\prime t^2 \exp\left[ -M \mathrm{e}^{-2-4c(1+\mathrm{e})t} \right].
\end{equation}
Using the definition of $\lambda(t)$ in (\ref{eq:lambdadef}) and the concavity of the logarithm, along with $\log x < x$: \begin{equation}
\lambda(t) \le 2c \left[ 1 + \mathrm{e}(M+1)N^\prime c t \exp\left[ -M \mathrm{e}^{-2-4c(1+\mathrm{e})t} \right] \right].
\end{equation}

Let us define \begin{equation}
t = \frac{\log M - 2 - r}{4c(1+\mathrm{e})}.
\end{equation}
Then, using $M+1 \le 2N^\prime$ and $\log M \le N^\prime$: \begin{equation}
\frac{\lambda}{2c} \le 1 + \frac{\mathrm{e} N^{\prime 3} }{2(1+\mathrm{e})} \exp\left[- \mathrm{e}^r \right] < 1 + \frac{ N^{\prime 3} }{2} \exp\left[- \mathrm{e}^r \right].
\end{equation}
Demanding that the inequality in (\ref{eq:lambdadef}) holds and solving for $r$, we obtain (\ref{eq:scramblingtime}).
\end{proof}

This theorem can be interpreted as follows.   For any $0<\kappa<1$, define the operator scrambling time \begin{equation}
t_{\mathrm{s},\kappa} = \inf \left\lbrace t\in \mathbb{R}^+ : \mathbb{E}_{\mathrm{s},t}[l] \ge \kappa N^\prime  \right\rbrace .
\end{equation} 
(Recall that $N^\prime$ is the maximal value of $l$, as defined in (\ref{eq:Nprimedef})).  It was conjectured in \cite{sekino} that a quantum ``scrambling time" $t_{\mathrm{s}} = \mathrm{\Omega}(\log N)$ would necessarily grow at least logarithmically with the number of degrees of freedom in any system with few-body interactions.   For example, in the SYK model, we would demand that $q$ is finite.  In recent years, this operator scrambling time has become the preferred definition of scrambling in the physics literature, though this is likely out of convenience \cite{lucas1805}.
Theorem \ref{lyapunovtheor} implies that $t_{\mathrm{s},\kappa} = \mathrm{O}(\log N)$, as summarized in the following corollary:
\begin{cor} \label{corolscramble}
If (\ref{eq:Kll}) holds for $M = N^\alpha$ for $0<\alpha<1$, then there exists an $N$-independent $b \in \mathbb{R}^+$ for which \begin{equation}
t_{\mathrm{s},\kappa } \ge b \log N \label{eq:scrambleb}
\end{equation}
for $N>N_0$, for some finite $N_0\in \mathbb{Z}^+$.
\end{cor}
\begin{proof}
There exists an $N_*\in\mathbb{Z}^+$ such that \begin{equation}
\frac{\alpha \log N_* }{8c(1+\mathrm{e})} < \frac{1}{4c(1+\mathrm{e})} \left[ \alpha \log N_* - 2 -\log\log N_*^{\prime 3}\right].
\end{equation}
Suppose that $N>N_*$.   Using Theorem \ref{lyapunovtheor}, we conclude that at time $t=t_{\mathrm{s}}$, where $t_{\mathrm{s}}$ is given by (\ref{eq:scrambleb}) where  \begin{equation}
b :=\frac{\alpha }{8c(1+\mathrm{e})} ,
 \end{equation}
 \begin{equation}
 \mathbb{E}_{\mathrm{s}}[l] \le \exp\left[ \frac{3\alpha \log N}{8(1+\mathrm{e})}\right] = N^{3\alpha/8(1+\mathrm{e})}
 \end{equation}
 We conclude that the corollary holds so long as $N_0$ is chosen such that $\kappa N_0 > N_0^{3\alpha/8(1+\mathrm{e})}$ and $N_0 \ge N_*$.
\end{proof}

We emphasize that the results of this section are completely general, and apply to a large family of models beyond the SYK model, as soon as (\ref{eq:Kll}) can be proved.

\section{Operator growth in the SYK ensemble}

\subsection{Bounding the transition rates}\label{sec:trans}
What remains is to show that (\ref{eq:Kll}) holds in the SYK ensemble, with very high probability, at large $N$.   Proving this fact constitutes the second main result of this paper.   The result is summarized in the following theorem:
\begin{thm}\label{theorSYK}
Let $\kappa \in \mathbb{R}^+$ and $\theta \in \mathbb{R}^+$ obey \begin{subequations}\begin{align}
2\kappa \log N + 2 &< \sqrt{N},   \label{eq:kappasqrtNN} \\
2(q-2) &< N^{\kappa\theta} - 1, \\
2q (1+\sqrt{N})&<(q\kappa \log N - 1) \sqrt{N} . \label{eq:annoying} 
\end{align}\end{subequations}
Let us also assume that \begin{equation}
q < \frac{N}{2}.  \label{eq:qlessN2}
\end{equation}Then, in the SYK model introduced in Section \ref{sec:ensemble}, with probability at least \begin{equation}
\mathbb{P}_{\mathrm{success}} \ge 1 - \frac{2(q-2)}{N^{\kappa \theta}-1},
\end{equation}
(\ref{eq:Kll}) is obeyed with 
\begin{subequations}\label{eq:theorcm}\begin{align}
 c &= \mathrm{e}^{\theta + 1/\kappa} \left[\sqrt{\frac{2(q-2)}{q}} \left(1-\frac{2\theta}{5\kappa \sqrt{N}\log N}\right)^{-1} + \frac{8}{q^{q-9/2}N^{1/4}} \left(\frac{4\theta}{5\kappa \log N} \right)^{(q-3)/2}  \right] , \\
 M &= \left\lfloor  \frac{\theta}{5\kappa} \frac{\sqrt{N}}{q^3\log N} \right\rfloor -1 .
\end{align}\end{subequations}
\end{thm}
\begin{proof}
Our strategy will be to work primarily with $\mathcal{K}_{s^\prime s}$.   At the very end of the calculation, we will use (\ref{eq:Kldef}) to bound $\mathcal{K}_l$.   We begin with the following proposition:
\begin{prop}
Define the symmetric and positive semidefinite matrix \begin{equation}
M_{s^\prime s} = \mathbb{Q}_s \mathcal{L}^{\mathsf{T}} \mathbb{Q}_{s^\prime}\mathcal{L}\mathbb{Q}_s.  \label{eq:MssDef}
\end{equation}
If the maximal eigenvalue of $M_{s^\prime s}$ is $\mu_{s^\prime s}$, \begin{equation}
\mathcal{K}_{s^\prime s} = \sqrt{\mu_{s^\prime s}}. \label{eq:Kmu}
\end{equation} \label{propmuss}
\end{prop}
\begin{proof}
 Let $\mathcal{O}\in \mathcal{B}$ obey $\lVert \mathcal{O}\rVert = 1$, and define \begin{equation}
|\mathcal{O}^\prime) = \mathbb{Q}_{s^\prime}\mathcal{L}\mathbb{Q}_s|\mathcal{O}).
\end{equation} From (\ref{eq:endBnorm}) and (\ref{eq:calKss}), we see that $\mathcal{K}_{s^\prime s}$ is simply the maximal length of the vector $|\mathcal{O}^\prime)$.  Now observe that $M_{s^\prime s}$ gives us a very simple way of measuring the length of $|\mathcal{O}^\prime)$.  Therefore,
 \begin{equation}
\mathcal{K}_{s^\prime s}^2 = \sup_{\mathcal{O}\in\mathcal{B}} (\mathcal{O}^\prime |  \mathcal{O}^\prime ) = \sup_{\mathcal{O}\in\mathcal{B}} (\mathcal{O}|M_{s^\prime s}|\mathcal{O}) = \mu_{s^\prime s},
\end{equation}
where for the last equality we used a variational principle which holds for a symmetric matrix.
\end{proof}

Denote \begin{equation}
C_s := \frac{N!}{s!(N-s)!},
\end{equation}
and observe that $M_{s^\prime s} \in \mathbb{R}^{C_s\times C_s}$ is a positive semidefinite random matrix.   From Markov's inequality, for any $p\in\mathbb{Z}^+$, \cite{furedi} 
\begin{equation}
\mathbb{P}\left[\mu_{s^\prime s} \ge a \right] \le \frac{\mathbb{E}\left[ \mu_{s^\prime s}^p\right]}{a^p}\le \frac{\mathbb{E}\left[ \mathrm{tr}(M_{s^\prime s}^p)\right]}{a^p}.  \label{eq:markov4}
\end{equation}
We will choose $p=\mathrm{O}(\log C_s)$, so that the number of eigenvalues $C_s$ accounted for in the trace is irrelevant ($C_s^{1/p} \rightarrow 1$).    Importantly, at finite size $s$, $p = \mathrm{O}(s\log N)$.   We will see that this is sufficiently small to make bounding $\mu_{s^\prime s}$ analytically tractable.

Hence, let us define \begin{equation}
B_{s^\prime s}^{(p)} := \mathbb{E}\left[\mathrm{tr}\left(M^p_{s^\prime s}\right)\right]. \label{eq:BssDef}
\end{equation}
We analyze the average $\mathbb{E}[\cdots]$ over the random variables $J_X$ by converting it to a combinatorial problem.   To do so, let us write out \begin{equation}
B^{(p)}_{s^\prime s} = \mathbb{E}\left[ \sum_{X_1,\ldots, X_p, Y_1,\ldots,Y_p \in F} \sum_{Z\subseteq V} (Z| \prod_{i=1}^p \mathbb{Q}_s \mathcal{L}^{\mathsf{T}}_{X_i}\mathbb{Q}_{s^\prime}\mathcal{L}_{Y_i}\mathbb{Q}_s |Z)   \right] \label{eq:explicittrace}
\end{equation}
where the sum over $Z$ is a sum over the basis of Proposition \ref{prop1}, without loss of generality.  We now read (\ref{eq:explicittrace}) from right to left, starting with $\mathbb{Q}_s|Z)$, which restricts the subset $Z\subseteq V$ to have exactly $s$ elements: $|Z|=s$, and $Z=\lbrace i_1,i_2,\ldots, i_s\rbrace$.  We draw a graph $G$ which we associate to $\mathbb{Q}_s |Z)$: \begin{equation} \label{eq:minimalgraph}
\begin{tikzpicture}
\draw (-0.5,0) node[left] {$\mathbb{Q}_s |i_1\cdots i_s) \sim$};
\draw (1.15, -0.4) -- (0.3, 0.25);
\draw (1.15, -0.4) -- (0.8, 0.25);
\draw (1.15, -0.4) -- (2, 0.25);
\fill[color=orange] (1,-0.5) -- ++(60:0.25) -- ++(-60:0.25) -- cycle;
\fill[color=blue] (0.3, 0.25) circle (3pt);
\fill[color=blue] (0.8, 0.25) circle (3pt);
\fill[color=blue] (2, 0.25) circle (3pt);
\draw (1.4, 0.25) node {$\cdots$};
\draw (0.3, 0.35) node[above] {\color{blue} \footnotesize $i_1$};
\draw (0.8, 0.35) node[above] {\color{blue} \footnotesize $i_2$};
\draw (2, 0.35) node[above] {\color{blue} \footnotesize $i_s$};
\end{tikzpicture}
\end{equation}
where the (blue) circles denote the fermions, and the orange triangle is a ``root" to the graph -- it has edges drawn to the $s$ original fermions in the operator $\mathbb{Q}_s |Z)$ (recall that $|Z)$ corresponds to a product operator).  We have written a $\sim$ in (\ref{eq:minimalgraph}) because we will not bother to keep track of an overall sign in the vector, although its orientation in $\mathcal{B}$ and its dependence on any random variables $J_X$ are each important. For simplicity, let us assume that $s^\prime=s+q-2$.  Without loss of generality, we assume that $Y_p \cap Z = \lbrace i_s\rbrace$,  $Y=\lbrace i_s,j_1,\ldots, j_{q-1}\rbrace$.   Then, we draw 
\begin{equation} \label{eq:graphtree}
\begin{tikzpicture}
\draw (-0.5, 0.3)  node[left] {$\mathbb{Q}_{s^\prime}\mathcal{L}_{Y_p}\mathbb{Q}_s|Y)\sim$};
\draw (1.15, -0.4) -- (0.3, 0.25);
\draw (1.15, -0.4) -- (0.8, 0.25);
\draw (1.15, -0.4) -- (2, 0.25) -- (3, 0.25) -- (2.2, 0.85);
\draw (2.7, 0.85) -- (3,0.25) -- (3.7, 0.85);
\fill[color=orange] (1,-0.5) -- ++(60:0.25) -- ++(-60:0.25) -- cycle;
\fill[color=blue] (0.3, 0.25) circle (3pt);
\fill[color=blue] (0.8, 0.25) circle (3pt);
\fill[color=blue] (2, 0.25) circle (3pt);
\fill[color=blue] (2.2, 0.85) circle (3pt);
\fill[color=blue] (2.7, 0.85) circle (3pt);
\fill[color=blue] (3.7, 0.85) circle (3pt);
\fill[color=red] (2.9, 0.15) rectangle (3.1, 0.35);
\draw (1.4, 0.25) node {$\cdots$};
\draw (3.2, 0.85) node {$\cdots$};
\draw (0.3, 0.35) node[above] {\color{blue} \footnotesize $i_1$};
\draw (0.8, 0.35) node[above] {\color{blue} \footnotesize $i_2$};
\draw (2, 0.35) node[above] {\color{blue} \footnotesize $i_s$};
\draw (2.2, 0.95) node[above] {\color{blue} \footnotesize $j_1$};
\draw (2.7, 0.95) node[above] {\color{blue} \footnotesize $j_2$};
\draw (3.7, 0.95) node[above] {\color{blue} \footnotesize $j_{q-1}$};
\draw (3, 0.15) node[below] {\color{red} \footnotesize $Y_p$};
\end{tikzpicture}.
\end{equation}
The way to read this graph is as follows:  the coupling (factor, drawn as a red square) $Y_p$ connected to the fermion $i_s$, and spawned the fermions $j_1,\ldots, j_{q-1}$.   Each fermion (circle) with an odd degree is present in the operator; those with an even degree are not present.   Because of the projectors $\mathbb{Q}_{s^\prime}$ and $\mathbb{Q}_s$, we had to start with an operator of size $s$ and add exactly $q-2$ net fermions.    From (\ref{eq:psiApsiB}), we know that  this vector is proportional to one of our simple basis vectors (a product operator), which is why we can simply draw the graph (so long as we neglect the proportionality coefficient).  The fermions do not directly connect to each other, but rather connect through the factors.   

Let us continue and study the operator $\mathbb{Q}_s \mathcal{L}^{\mathsf{T}}_{X_i}\mathbb{Q}_{s^\prime}\mathcal{L}_{Y_i}\mathbb{Q}_s |Z)$.  It is easiest to first illustrate the possibilities with a simple example.  Consider the theory with $s=3$, $s^\prime=5$, $q=4$.  Let us first consider the theory where $X_p=Y_p = \lbrace 3,4,5,6\rbrace$ and $Z=\lbrace 1,2,3\rbrace$.   Then we draw 
\begin{equation}\label{eq:Xpexample1}
\begin{tikzpicture}
\draw (-0.5,0) node[left] {$\mathbb{Q}_s \mathcal{L}^{\mathsf{T}}_{X_p}\mathbb{Q}_{s^\prime}\mathcal{L}_{Y_p}\mathbb{Q}_s |Z)\sim$};
\draw (1.15, -0.4) -- (0.3, 0.25);
\draw (1.15, -0.4) -- (1.15, 0.25);
\draw (1.15, -0.4) -- (2, 0.25);
\fill[color=orange] (1,-0.5) -- ++(60:0.25) -- ++(-60:0.25) -- cycle;
\fill[color=blue] (0.3, 0.25) circle (3pt);
\fill[color=blue] (1.15, 0.25) circle (3pt);
\fill[color=blue] (2, 0.25) circle (3pt);
\draw (0.3, 0.35) node[above] {\color{blue} \footnotesize $1$};
\draw (1.15, 0.35) node[above] {\color{blue} \footnotesize $2$};
\draw (2, 0.35) node[above] {\color{blue} \footnotesize $3$};
\end{tikzpicture}
\end{equation}
where the absence of the factor $Y_p$ reminds us that since $J_{Y_p}$ has appeared twice in the sequence, this sequence is non-trivial under random averaging.   We neglect to draw any fermion or factor which has degree zero, which is why the fermions 4, 5 and 6 are not shown.   

However, suppose instead $X_p = \lbrace 2,4,5,7\rbrace$.  In this case, 
\begin{equation}\label{eq:Xpexample2}
\begin{tikzpicture}
\draw (-0.5,0.3) node[left] {$\mathbb{Q}_s \mathcal{L}^{\mathsf{T}}_{X_p}\mathbb{Q}_{s^\prime}\mathcal{L}_{Y_p}\mathbb{Q}_s |Z)\sim$};
\draw (1.15, -0.4) -- (0.3, 0.25);
\draw (1.15, -0.4) -- (1.15, 0.25) -- (1.8, 1.1) -- (1, 1.1);
\draw (1.15, -0.4) -- (2, 0.25) -- (3, 0.25) -- (2.3, 0.85);
\draw (3, 0.85) -- (3,0.25) -- (3.7, 0.85);
\draw (3, 0.85) -- (1.8, 1.1) -- (2.3, 0.85);
\fill[color=orange] (1,-0.5) -- ++(60:0.25) -- ++(-60:0.25) -- cycle;
\fill[color=blue] (0.3, 0.25) circle (3pt);
\fill[color=blue] (1.15, 0.25) circle (3pt);
\fill[color=blue] (2, 0.25) circle (3pt);
\fill[color=blue] (2.3, 0.85) circle (3pt);
\fill[color=blue] (3, 0.85) circle (3pt);
\fill[color=blue] (3.7, 0.85) circle (3pt);
\fill[color=blue] (1,1.1) circle (3pt);
\fill[color=red] (2.9, 0.15) rectangle (3.1, 0.35);
\fill[color=red] (1.7, 1) rectangle (1.9, 1.2);
\draw (0.3, 0.35) node[above] {\color{blue} \footnotesize $1$};
\draw (1.15, 0.35) node[above] {\color{blue} \footnotesize $2$};
\draw (2, 0.35) node[above] {\color{blue} \footnotesize $3$};
\draw (2.3, 0.95) node[above] {\color{blue} \footnotesize $4$};
\draw (3, 0.95) node[above] {\color{blue} \footnotesize $5$};
\draw (3.7, 0.95) node[above] {\color{blue} \footnotesize $6$};
\draw (1,1.2) node[above] {\color{blue} \footnotesize $7$};
\draw (1.8, 1.2) node[above] {\color{red} \footnotesize $X_p$};
\draw (3, 0.15) node[below] {\color{red} \footnotesize $Y_p$};
\end{tikzpicture}
\end{equation}
Because the factors $X_p \ne Y_p$, we must draw both of them, together with an edge to all vertices/fermions $i\in X_p$ or $Y_p$.   We can only remove a factor when that exact factor shows up a second time.   And if a factor shows up a third time, it is redrawn in, and so on.   Note that the only factors that $X_p$ can be are those which destroy 3 fermions and create one, in this simple example. 

It is straightforward to generalize these rules, which we summarize one more time.  If the next factor in the sequence is present in the existing factor graph,  that factor is deleted along with its edges to $q$ fermions.  Any fermions which subsequently have degree zero are removed.  If the factor is new, we draw that factor and $q$ edges to its fermions.  The number of odd degree fermions in each graph is fixed by the projectors to alternate between $s$ and $s^\prime$.  


Our next goal is to throw away detailed information about what specific factors and fermions appeared, and to only keep track of the sequence of graphs.  Let $\mathcal{G}$ be the space of all graphs, modulo graph isomorphism.  Two graphs $G_1$ and $G_2$ are isomorphic if and only if there is a permutation on fermions $\pi \in \mathrm{S}^V$ such that $\pi \cdot G_1 = G_2$ (the group action on fermions and factors is canonical, while the root is invariant).   We define $G_\triangle$ to be the unique (up to isomorphism) element of $\mathcal{G}$  with zero factors and $s$ fermions connected to the root, as in (\ref{eq:minimalgraph}).  Let $\mathcal{G}_-$ be the subset of $\mathcal{G}$ consisting of $s$ odd degree fermions and no more than $p$ factors, and $\mathcal{G}_+$ be the subset of $\mathcal{G}$ with $s^\prime$ odd degree fermions, subject to the constraint that any graph in $\mathcal{G}_+$ or $\mathcal{G}_-$ can be reached by adding and removing factors to $G_\triangle$, according to the rules above, and with no intermediate graphs containing more than $p$ factors.  

Define $\langle G_2| \mathcal{N}_+|G_1 \rangle$ to be the number of factors $X$ which can be added (or removed) to any fixed graph $G$ isomorphic to $G_1 \in \mathcal{G}_-$, to create any graph isomorphic to $G_2 \in \mathcal{G}_+$.   Similarly, we define  $\langle G_1| \mathcal{N}_-|G_2 \rangle$ to be the number of factors which take a fixed graph isomorphic to $G_2\in \mathcal{G}_+$ to any graph isomorphic to $G_1 \in \mathcal{G}_-$.   We interpret $\mathcal{N}_+ : \mathbb{Z}^{\mathcal{G}_+ \times \mathcal{G}_-} \rightarrow \mathbb{Z}$ and $\mathcal{N}_- : \mathbb{Z}^{\mathcal{G}_- \times \mathcal{G}_+} \rightarrow \mathbb{Z}$ as integer-valued matrices, using the angle bra-ket notation to denote matrix elements.   Many of these matrix elements are zero.  $\mathcal{N}_+$ and $\mathcal{N}_-$ are both non-negative matrices.

\begin{prop}
\label{propcomb}
If  $G_0 = G_p = G_\triangle$, then \begin{equation}
B_{s^\prime s}^{(p)} \le \left(4\sigma^2\right)^p C_s \sum_{ \substack{G_1,\ldots, G_{p-1} \in \mathcal{G}_- \\H_1,\ldots, H_p \in \mathcal{G}_+ } }\prod_{i=1}^p \langle G_i | \mathcal{N}_-|H_i\rangle\langle H_i|\mathcal{N}_+|G_{i-1}\rangle . \label{eq:combineq}
\end{equation}
\end{prop}

\begin{proof} We expand out the sums in (\ref{eq:explicittrace}) over all possible couplings.  Using the algorithm described above to associate a graph with each of the $2p+1$ operators $\mathbb{Q}_s|Z)$, $\mathbb{Q}_{s^\prime}\mathcal{L}_{Y_p}\mathbb{Q}_s|Z)$, etc., we may convert the sequence of factors \begin{equation}
\mathcal{Y}_Z := (Y_p, X_p, Y_{p-1},X_{p-1},\ldots, Y_1,X_1)_Z \label{eq:calYZ}
\end{equation}
read right to left, along with the initial  operator $|Z)$, into a sequence of $2p+1$ graphs \begin{equation}
\mathcal{Z}(\mathcal{Y}_Z) := (G_0, H_1, G_1, H_2,\ldots, H_p, G_p) \label{eq:calZdef}
\end{equation}
with $G_0=G_\triangle$.  Due to the projectors $\mathbb{Q}_s$ and $\mathbb{Q}_{s^\prime}$, any sequence $\mathcal{Y}$ which (before disorder averaging) is not zero must map to a sequence $\mathcal{Z}(\mathcal{Y}_Z)$ in which $G_i \in \mathcal{G}_-$ and $H_i \in \mathcal{G}_+$.  When calculating $B_{s^\prime s}^{(p)}$,  we require that $G_p = G_\triangle$;  otherwise, there is a coupling which appears an odd number of times in $\mathcal{Y}$, so the disorder average of that sequence vanishes.   We define \begin{equation}
\mathcal{Z}_{s^\prime s}^{(p)} := \lbrace (G_\triangle, H_1, G_1, H_2,\ldots, H_p, G_\triangle) : H_i \in \mathcal{G}_+,G_i \in \mathcal{G}_- \rbrace. 
\end{equation}   Since only one factor can be added or removed in each step, it is not possible to have more than $p$ factors in any graph in $\mathcal{Z}(\mathcal{Y}_Z)$.   We define the equivalence relation $\mathcal{Y}_1 \sim \mathcal{Y}_2$ if and only if $\mathcal{Z}(\mathcal{Y}_1) = \mathcal{Z}(\mathcal{Y}_2)$, and denote $\mathcal{Y}_{1} \in \mathcal{Z}(\mathcal{Y})$.

We then write \begin{equation}
B_{s^\prime s}^{(p)} = \sum_{Z\subseteq V} \sum_{\mathcal{Z} \in \mathcal{Z}_{s^\prime s}^{(p)}} \sum_{\mathcal{Y}_Z\in \mathcal{Z}}  (Z| \prod_{i=1}^p \mathbb{Q}_s \mathcal{L}^{\mathsf{T}}_{X_i}\mathbb{Q}_{s^\prime} \mathcal{L}_{Y_i}\mathbb{Q}_s |Z).
\end{equation}
Due to the Rademacher distribution on the random variables $J_X$, the expectation value has become trivial, encoded in the fact that the graph sequence $\mathcal{Z}$ ends at $G_\triangle$.  For other distributions on $J_X$, the sum above must be weighted in a more complicated way when the sum involves $\mathbb{E}[J_X^{2k}]$ for $k>1$.   We now apply the triangle inequality together with $(Z|Z)=1$: \begin{equation}
B_{s^\prime s}^{(p)} \le   \sum_{Z\subseteq V : |Z|=s} \sum_{\mathcal{Z} \in \mathcal{Z}_{s^\prime s}^{(p)}} \sum_{\mathcal{Y}_Z\in \mathcal{Z}} \prod_{i=1}^p \lVert \mathcal{L}_{X_i}\rVert  \lVert \mathcal{L}_{Y_i}\rVert =  \sum_{\mathcal{Z} \in \mathcal{Z}_{s^\prime s}^{(p)}}   \sum_{Z\subseteq V : |Z|=s} \sum_{\mathcal{Y}_Z\in \mathcal{Z}} \left(2\sigma\right)^{2p}.
\end{equation}
In the last step, we used (\ref{eq:psiApsiB}) along with the fact that we may exchange the first two sums, whose summands are independent.  It now remains to evaluate each sum in turn.   By definition of $\mathcal{N}_+$ and $\mathcal{N}_-$: \begin{equation}
\sum_{\mathcal{Y}_Z\in \mathcal{Z}} \left(2\sigma\right)^{2p} = \left(4\sigma^2\right)^p \prod_{i=1}^p \langle G_i |\mathcal{N}_-|H_i\rangle \langle H_i | \mathcal{N}_+ |G_{i-1}\rangle,
\end{equation}
using (\ref{eq:calZdef}) and $G_0=G_\triangle$.   By permutation symmetry, \begin{equation}
\sum_{Z\subseteq V : |Z|=s}\left(4\sigma^2\right)^p \prod_{i=1}^p \langle G_i |\mathcal{N}_-|H_i\rangle \langle H_i | \mathcal{N}_+ |G_{i-1}\rangle = \left(4\sigma^2\right)^pC_s \prod_{i=1}^p \langle G_i |\mathcal{N}_-|H_i\rangle \langle H_i | \mathcal{N}_+ |G_{i-1}\rangle
\end{equation}
Hence we obtain (\ref{eq:combineq}).
\end{proof}

\begin{prop}
\label{propsym}
Let $\mathcal{Z} = (G_\triangle, H_1, G_1, \ldots, G_{p-1}, H_p, G_\triangle) \in \mathcal{Z}_{s^\prime s}^{(p)}$.  Then  \begin{equation}
\prod_{i=1}^p \langle G_i | \mathcal{N}_-|H_i\rangle\langle H_i|\mathcal{N}_+|G_{i-1}\rangle = \prod_{i=1}^p \langle G_{i-1} | \mathcal{N}_-|H_i\rangle\langle H_i|\mathcal{N}_+|G_{i}\rangle .\label{eq:symid}
\end{equation}
\end{prop}
\begin{proof}
Pick any $Z\subseteq V$ with $|Z|=s$.  Given a sequence of factors $\mathcal{Y}_Z$, given by (\ref{eq:calYZ}), with $\mathcal{Y}_Z \in \mathcal{Z}$, define the reversed sequence \begin{equation}
\mathcal{Y}_Z^{\mathrm{r}} := (X_1,Y_1,\ldots, X_p, Y_p)_Z.
\end{equation}
which corresponds to factors in (\ref{eq:explicittrace}) read left to right instead.  By construction, $\mathcal{Y}_Z^{\mathrm{r}} \in \mathcal{Z}^{\mathrm{r}}$, defined by \begin{equation}
\mathcal{Z}^{\mathrm{r}} := (G_\triangle, H_p, G_{p-1}, \ldots, G_1, H_1, G_\triangle).
\end{equation}
Clearly, $\left(\mathcal{Y}_Z^{\mathrm{r}}\right)^{\mathrm{r}} = \mathcal{Y}_Z$ and $\left(\mathcal{Z}^{\mathrm{r}}\right)^{\mathrm{r}} = \mathcal{Z}$.  As each sequence $\mathcal{Y}_Z$ has a unique reverse $\mathcal{Y}_Z^{\mathrm{r}}$, \begin{equation}
\prod_{i=1}^p \langle G_i | \mathcal{N}_-|H_i\rangle\langle H_i|\mathcal{N}_+|G_{i-1}\rangle = \sum_{\mathcal{Y}_Z \in \mathcal{Z}} 1 = \sum_{\mathcal{Y}_Z^{\mathrm{r}} \in \mathcal{Z}^{\mathrm{r}}} 1 = \prod_{i=1}^p \langle G_{i-1} | \mathcal{N}_-|H_i\rangle\langle H_i|\mathcal{N}_+|G_{i}\rangle ,
\end{equation}
which completes the proof.
\end{proof}

\begin{lem}\label{lma12}
Define a transfer matrix $\mathcal{M}_{s^\prime s}^{(p)}\in \mathbb{R}^{\mathcal{G}_s^{(p)} \times \mathcal{G}_s^{(p)}} $ component wise as \begin{equation} 
\langle G_1|\mathcal{M}_{s^\prime s}^{(p)} |G_2\rangle = \sum_{H\in \mathcal{G}_{s^\prime}^{(p)}} \langle G_1|\mathcal{N}_-|H\rangle\langle H|\mathcal{N}_+|G_2\rangle.
\end{equation}
Then if $\nu_{s^\prime s}^{(p)}$ is the maximal (left or right) eigenvalue of $\mathcal{M}_{s^\prime s}^{(p)}$,
\begin{equation}
B_{s^\prime s}^{(p)} \le C_s \left(\nu_{s^\prime s}^{(p)}\right)^p.  \label{eq:Bfinalbound}
\end{equation}
\end{lem}
\begin{proof}
Rewriting (\ref{eq:combineq}) in terms of the transfer matrix: \begin{equation}
B_{s^\prime s}^{(p)} \le \langle G_\triangle | \left(\mathcal{M}_{s^\prime s}^{(p)} \right)^p |G_\triangle\rangle.
\end{equation}
Now, letting $G_0=G_p=G_\triangle$, and using the property that \begin{equation}
\sum_{H_1,\ldots, H_p \in \mathcal{G}_+}\prod_{i=1}^p \langle G_i | \mathcal{N}_-|H_i\rangle\langle H_i|\mathcal{N}_+|G_{i-1}\rangle =\sum_{H_1,\ldots, H_p \in \mathcal{G}_+} \prod_{i=1}^p \langle G_{i-1} | \mathcal{N}_-|H_i\rangle\langle H_i|\mathcal{N}_+|G_{i}\rangle \label{eq:presymmetrizing}
\end{equation} which follows from Proposition \ref{propsym},
\begin{align}
B_{s^\prime s}^{(p)} &\le \left(4\sigma^2\right)^p C_s  \sum_{ \substack{G_1,\ldots, G_{p-1} \in \mathcal{G}_s^{(p)} \\H_1,\ldots, H_p \in \mathcal{G}_{s^\prime}^{(p)} } }\prod_{i=1}^p \langle G_i | \mathcal{N}_-|H_i\rangle\langle H_i|\mathcal{N}_+|G_{i-1}\rangle \notag \\
&= C_s  \sum_{   G_1,\ldots, G_{p-1} \in \mathcal{G}_- }\prod_{i=1}^p \left[4\sigma^2 \sum_{H_i \in \mathcal{G}_+}  \langle G_i | \mathcal{N}_-|H_i\rangle\langle H_i|\mathcal{N}_+|G_{i-1}\rangle\right] \notag \\
&= C_s \sum_{   G_1,\ldots, G_{p-1} \in \mathcal{G}_- }\prod_{i=1}^p \left[4\sigma^2 \sqrt{\sum_{H_i \in \mathcal{G}_+}  \langle G_i | \mathcal{N}_-|H_i\rangle\langle H_i|\mathcal{N}_+|G_{i-1}\rangle} \sqrt{\sum_{H_i \in \mathcal{G}_+}  \langle G_{i-1} | \mathcal{N}_-|H_i\rangle\langle H_i|\mathcal{N}_+|G_{i}\rangle} \right] \notag \\
&= C_s \langle G_\triangle | \left(\widetilde{\mathcal{M}}_{s^\prime s}^{(p)} \right)^p |G_\triangle\rangle \label{eq:symmetrizing}
\end{align}
where we have defined the symmetrized transfer matrix \begin{equation}
\langle G_1|\widetilde{\mathcal{M}}_{s^\prime s}^{(p)} |G_2\rangle := \sqrt{\langle G_1|\mathcal{M}_{s^\prime s}^{(p)} |G_2\rangle\langle G_2|\mathcal{M}_{s^\prime s}^{(p)} |G_1\rangle}.
\end{equation}
In the third line of (\ref{eq:symmetrizing}), we have used the distributive property along with (\ref{eq:presymmetrizing}) and the trivial identity that $x=y$ implies $\sqrt{xy}=x$.

Since $\mathcal{M}_{s^\prime s}^{(p)}$ is non-negative, $\widetilde{\mathcal{M}}_{s^\prime s}^{(p)}$ is a symmetric and positive semidefinite and acts on a finite dimensional vector space.  Let $\widetilde{\nu}_{s^\prime s}^{(p)}$ be its maximal eigenvalue.  We conclude (for example, using elementary variational methods) that, since $\langle G_\triangle | G_\triangle \rangle = 1$, \begin{equation}
 \langle G_\triangle | \left(\widetilde{\mathcal{M}}_{s^\prime s}^{(p)} \right)^p |G_\triangle\rangle \le \left(\widetilde{\nu}_{s^\prime s}^{(p)}\right)^p. \label{eq:nutildebound}
\end{equation}

It remains to relate $\widetilde{\nu}_{s^\prime s}^{(p)}$ to $\nu_{s^\prime s}^{(p)}$.  By definition of the set $\mathcal{G}_-$ (and the fact it has a finite number of elements), for any two graphs $G_{1,2} \in \mathcal{G}_-$, there exists an integer $n<\infty$ such that \begin{equation}
\langle G_1 | \left(\mathcal{M}_{s^\prime s}^{(p)}\right)^n |G_2\rangle > 0.
\end{equation}
This identity follows from the fact that there exist a sequence of factors from $G_\triangle$ to $G_1$ and $G_2$, as well as the reverse sequences from $G_2$ or $G_1$ to $G_\triangle$.  Hence $\mathcal{M}_{s^\prime s}^{(p)}$ is an irreducible non-negative matrix, and it follows that $\widetilde{\mathcal{M}}_{s^\prime s}^{(p)}$ is also irreducible.  By the Perron-Frobenius Theorem \cite{meyer}: (\emph{1}) $\mathcal{M}_{s^\prime s}^{(p)}$ has a maximal eigenvector $|\phi \rangle$ and $\widetilde{\mathcal{M}}_{s^\prime s}^{(p)}$ has a maximal eigenvector $|\widetilde{\phi} \rangle$, which each obey \begin{equation}
\langle G_\triangle|\phi\rangle \ne 0 \text{ and } \langle G_\triangle|\widetilde{\phi}\rangle \ne 0; \label{eq:phiGne0}
\end{equation}
(\emph{2}) $\nu_{s^\prime s}^{(p)}$ and $\widetilde{\nu}_{s^\prime s}^{(p)}$ are non-degenerate; (\emph{3}) as stated in the lemma, $\nu_{s^\prime s}^{(p)}$ is the maximal left and maximal right eigenvalue of $\mathcal{M}_{s^\prime s}^{(p)}$.  As $\mathbb{R}^{\mathcal{G}_-}$ is a finite dimensional vector space, (\ref{eq:phiGne0}) implies that
\begin{align}
\nu_{s^\prime s}^{(p)} = \lim_{n\rightarrow \infty} \frac{\log \langle G_\triangle | \left(\mathcal{M}_{s^\prime s}^{(p)}\right)^n |G_\triangle\rangle }{n}= \lim_{n\rightarrow \infty} \frac{\log \langle G_\triangle | \left(\widetilde{\mathcal{M}}_{s^\prime s}^{(p)}\right)^n |G_\triangle\rangle }{n} = \widetilde{\nu}_{s^\prime s}^{(p)}.  \label{eq:nunutilde}
\end{align}
Combining (\ref{eq:nutildebound}) and (\ref{eq:nunutilde}) we obtain (\ref{eq:Bfinalbound}).
\end{proof}

Using Lemma \ref{lma12}, we now begin to bound the maximal eigenvalue of the matrix $M_{s^\prime s}$ defined in (\ref{eq:MssDef}).   Let \begin{equation}
p = \lceil \kappa s \log N \rceil,  \label{eq:pkappa}
\end{equation}
where the parameter $\kappa \in \mathbb{R}^+$ will be O(1).  We first combine (\ref{eq:markov4}) and (\ref{eq:BssDef}).    Using the inequality \begin{equation}
C_s > \left(\frac{\mathrm{e}N}{s}\right)^s,
\end{equation}
we find that since $N\ge 4$ and $s\ge 1$: \begin{equation}
\mathbb{P}\left[\mu_{s^\prime s}> \nu_{s^\prime s}^{(p)} \mathrm{e}^{\theta + 2/\kappa}\right] \le \left(\frac{\mathrm{e}^{-2/\kappa}}{\mathrm{e}^\theta}  \right)^p \left(\frac{\mathrm{e}N}{s}\right)^s = \left(\frac{\mathrm{e}^{-2/\kappa}}{\mathrm{e}^\theta} \left(\frac{\mathrm{e}N}{s}\right)^{1/\kappa \log N} \right)^p  \le \mathrm{e}^{-p\theta} = \frac{1}{N^{\kappa\theta s}}.
\end{equation}
Moreover, since there are at most $2(q-2)$ non-vanishing $\mathcal{K}_{s^\prime s}$ coefficients involving a fixed operator size, we conclude that \begin{align}
\mathbb{P}\left[\mu_{s^\prime s}> \nu_{s^\prime s}^{(p)} \mathrm{e}^{\theta + 2/\kappa}, \text{ for any } s, s^\prime \right] &\le \sum_{|s^\prime -s| \le q-2} \mathbb{P}\left[\mu_{s^\prime s}> \nu_{s^\prime s}^{(p)} \mathrm{e}^{\theta + 2/\kappa}\right] \le    2(q-2) \sum_{s=1}^N \frac{1}{N^{\kappa\theta s}} \notag \\
&< \frac{2(q-2)}{N^{\kappa\theta}-1} =1- \mathbb{P}_{\mathrm{success}}.
\end{align}
Hence, with probability $\mathbb{P}_{\mathrm{success}}$, we may assume that $\mu_{s^\prime s} \le \mathrm{e}^{\theta + 2/\kappa} \nu_{s^\prime s}^{(p)}$.

Of course, it remains to bound $\nu_{s^\prime s}^{(p)}$, which we do in the following lemma:

\begin{lem}\label{lma13}
If $p$ is given by (\ref{eq:pkappa}), $k$ is defined in (\ref{eq:kdef}), and we assume (\ref{eq:kappasqrtNN}), then \begin{equation}
\nu_{s^\prime s}^{(p)} < \left( \frac{s+q-2}{q-1} \frac{2s}{q} \left(\frac{2q^2s}{N}\right)^{2k-2}+ \frac{2^q (q-1)! (s+q)^{q-1}}{N^{(q-2)/2}} \right) \exp\left[ \frac{5q^2 \kappa s \log N}{\sqrt{N}}\right].  \label{eq:lma13}
\end{equation}

\end{lem}
\begin{proof}
Let us interpret $\nu_{s^\prime s}^{(p)}$ as the maximal left eigenvalue of $\mathcal{M}_{s^\prime s}^{(p)}$.  Defining $\mathbb{R}^{\mathcal{G}_-}_+$ as the set of all vectors with strictly positive entries, we begin by invoking the Collatz-Wielandt bound \cite{meyer}: \begin{equation}
\nu_{s^\prime s}^{(p)} = \inf_{|\phi\rangle \in \mathbb{R}^{\mathcal{G}_-}_+} \sup_{G\in \mathcal{G}_-} \frac{\langle \phi| \mathcal{M}_{s^\prime s}^{(p)} | G\rangle}{\langle \phi|G\rangle}.  \label{eq:collatzwieland}
\end{equation}
Clearly, we can bound $\nu_{s^\prime s}^{(p)}$ by simply guessing any $|\phi\rangle \in \mathbb{R}^{\mathcal{G}_-}_+$.   We choose  \begin{equation}
\langle \phi |G\rangle = N^{-|V\cap G|/2}
\end{equation}
where $|V\cap G|$ denotes the number of fermions (of non-zero degree!) in the graph $G$.   

Now let us write out \begin{equation}
\sup_{G\in \mathcal{G}_-} \frac{\langle \phi| \mathcal{M}_{s^\prime s}^{(p)} | G\rangle}{\langle \phi|G\rangle} = \sup_{G\in \mathcal{G}_-}  4\sigma^2\sum_{H^\prime \in \mathcal{G}_+, H\in \mathcal{G}_-} N^{(|G\cap V| - |H\cap V|)/2} \langle H|\mathcal{N}_-|H^\prime\rangle\langle H^\prime | \mathcal{N}_+|G\rangle .  \label{eq:supGsum}
\end{equation}
Given graphs $G$, $H$ and $H^\prime$, let us define the following four parameters: \begin{subequations}\begin{align}
a_+ := |(H^\prime - H^\prime \cap G)\cap V|, \\ 
a_- := |(H - H^\prime \cap H)\cap V|, \\ 
b_+ := |(G - H^\prime \cap G)\cap V|, \\ 
b_- := |(H^\prime - H^\prime \cap H)\cap V|.
\end{align}\end{subequations}
$a_+$ and $a_-$ are the number of \emph{new} fermions added to the graph in the first and second step respectively; $b_+$ and $b_-$ represent the number of fermions \emph{removed} from the graph in the first and second step respectively.   Note the following constraints: \begin{subequations}\label{eq:abbounds}\begin{align}
0\le a_+ \le q+1-2k, \\
0 \le a_- \le 2k-1, \\
0 \le b_+ \le 2k-1, \\
0 \le b_- \le q+1-2k.
\end{align}\end{subequations} \begin{equation}
|G\cap V| - |H\cap V| = b_+ + b_- - a_+ - a_-.  \label{eq:GVHVab}
\end{equation}
Lastly, note that $a_+$ and $a_-$ are non-negative if and only if a factor is added to the graph, and $b_+$ and $b_-$ are non-negative if and only if a factor is removed from the graph in that step.

There are four possible kinds of sequences of $H$ and $H^\prime$, corresponding to whether a factor is added (A) or removed (R) from the graph in each step: RR, AA, RA, AR.  Because we keep the starting graph fixed, and sum over all possible ways to add or remove factors to the graph, we can efficiently overestimate the sum over all possible modifications to the graph with fixed $a_\pm$ or $b_\pm$.   Let \begin{equation}
v=|G\cap V|
\end{equation}
to be the number of vertices in $G$.  In the first step, the number of ways to add a factor is \begin{equation}
N_{\mathrm{A}}(a_+) := \sum_{H^\prime \in \mathcal{G}_+ : |H^\prime \cap V| = a_+ + |G\cap V|} \langle H^\prime|\mathcal{N}_+|G\rangle =   \left(\begin{array}{c} N - v \\ a_+ \end{array}\right) \left(\begin{array}{c} s \\ 2k-1 \end{array}\right) \left(\begin{array}{c} v-s \\ q+1-2k-a_+ \end{array}\right),
\end{equation}
where the first combinatorial factor is the choice of $a_+$ distinct fermions to add to the graph, the second  is the number of $(2k-1)$-tuples of the $s$ odd degree fermions present in the graph, and the third is the number $(q+1-2k-a_+)$-tuples of even degree fermions to add an extra edge to.   If instead we remove a factor, we find \begin{equation}
N_{\mathrm{R}}(b_+) \le  \left\lbrace \begin{array}{ll} \displaystyle \left\lfloor \dfrac{s}{b_+} \right\rfloor &\ b_+ > 0 \\ p &\ b_+ = 0 \end{array}\right.,
\end{equation}
where the first line corresponds to the maximal number of factors that can have $b_+ > 0$ odd degree fermions, and the second line is a crude bound: we can remove no more factors than the maximal number $p$ allowed in any graph in $\mathcal{G}_-$.

Next we look at the AA sequences, where two factors are added sequentially.  Here we find \begin{equation}
N_{\mathrm{AA}}(a_+, a_-) := N_{\mathrm{A}}(a_+)  \left(\begin{array}{c} N - v - a_+ \\ a_- \end{array}\right) \left(\begin{array}{c} s+q+2-4k \\ q+1-2k \end{array}\right) \left(\begin{array}{c} v+a_+-s-q-2+4k \\ 2k-1-a_- \end{array}\right).
\end{equation}
In the RA sequences, we find \begin{equation}
N_{\mathrm{RA}}(b_+, a_-) \le N_{\mathrm{R}}(b_+) \times \left(\begin{array}{c} N - v + b_+ \\ a_- \end{array}\right) \left(\begin{array}{c} s+q+2-4k \\ q+1-2k \end{array}\right) \left(\begin{array}{c} v-b_+-s-q-2+4k \\ 2k-1-a_- \end{array}\right).
\end{equation}
For the RR sequences, we find \begin{equation}
N_{\mathrm{RR}}(b_+,b_-) \le N_{\mathrm{R}}(b_+) \times  \left\lbrace \begin{array}{ll} \displaystyle \left\lfloor \dfrac{s+q+2-4k}{b_-} \right\rfloor &\ b_- > 0 \\ p-1 &\ b_- = 0 \end{array}\right.,
\end{equation}
while for the AR sequences: \begin{equation}
N_{\mathrm{AR}}(a_+,b_-) \le N_{\mathrm{A}}(a_+) \times  \left\lbrace \begin{array}{ll} \displaystyle \left\lfloor \dfrac{s+q+2-4k}{b_-} \right\rfloor &\ b_- > 0 \\ p &\ b_- = 0 \end{array}\right..
\end{equation}

Now we must perform the sum over $a_\pm$ and $b_\pm$ in (\ref{eq:supGsum}).    We start with the sum over AA sequences, where we will crudely bound the six distinct choose functions for convenience.  Using (\ref{eq:GVHVab}), \begin{align}
\sum_{a_+=0}^{q+1-2k} \sum_{a_-=0}^{2k-1} \frac{N_{\mathrm{AA}}(a_+, a_-) }{N^{(a_+ + a_-)/2} }&< \sum_{a_+=0}^{q+1-2k} \sum_{a_-=0}^{2k-1} \frac{N^{a_+/2}  s^{2k-1} v^{q+1-2k-a_+}}{a_+! (2k-1)!(q+1-2k-a_+)!} \times \frac{N^{a_-/2} (s+q)^{q+1-2k} (v+q)^{2k-1-a_-}}{a_-! (q+1-2k)! (2k-1-a_-)!} \notag \\
&< \frac{(s+q)^q}{(q+1-2k)!(2k-1)!}\frac{\left(\sqrt{N}+v\right)^{q+1-2k}}{(q+1-2k)!} \frac{\left(\sqrt{N}+v+q\right)^{2k-1}}{(2k-1)!} \notag \\
&<  \frac{(s+q)^q \left(\sqrt{N}+v+q\right)^q }{(q+1-2k)!^2(2k-1)!^2}. \label{eq:lastsum1}
\end{align}
Next, we bound the RA sequences.  For convenience, we may just use $N_{\mathrm{R}}(b_+) \le p$: \begin{align}
\sum_{b_+=0}^{2k-1} \sum_{a_-=0}^{2k-1} \frac{N_{\mathrm{RA}}(a_+, a_-) }{N^{( a_- - b_+)/2} }&< \sum_{b_+=0}^{2k-1} \sum_{a_-=0}^{2k-1} p N^{b_+/2} \frac{N^{a_-/2} (s+q)^{q+1-2k} v^{2k-1-a_-} }{a_-! (q+1-2k)! (2k-1-a_-)!} \notag \\
&< \frac{p N^{(2k-1)/2}}{1-N^{(1-2k)/2}} \frac{(s+q)^{q+1-2k}}{(q+1-2k)!} \frac{\left(\sqrt{N} + v\right)^{2k-1}}{(2k-1)!} \notag \\
&< \frac{p (s+q)^{q+1-2k}}{(1-N^{-1/2})(q+1-2k)!(2k-1)!} \left(N + v\sqrt{N}\right)^{q/2} \label{eq:lastsum2}
\end{align}
Similarly, \begin{align}
\sum_{b_+=0}^{2k-1} \sum_{b_-=0}^{q+1-2k} N^{(b_+ + b_-)/2} N_{\mathrm{RR}}(b_+, b_-) < p(p-1) \frac{N^{q/2}}{(1-N^{-1/2})^2} \label{eq:lastsum3}
\end{align}
Lastly, and most importantly, we bound the AR sequences.  In this sum, we will split off the $b_- = q+1-2k$ contribution, and bound that more carefully: \begin{align}
\sum_{b_-=0}^{q+1-2k}  \sum_{a_+=0}^{q+1-2k} \frac{N_{\mathrm{AR}}(a_+, b_-)}{N^{(a_+ - b_-)/2}} &<  \left(p + \sum_{b_-=1}^{q+1-2k}   N^{b_-/2} \left\lfloor \frac{s+q+2-4k}{b_-} \right\rfloor \right)\left( \sum_{a_+=0}^{q+1-2k} \frac{N^{a_+/2}  s^{2k-1} v^{q+1-2k-a_+}}{a_+! (2k-1)!(q+1-2k-a_+)!} \right) \notag \\
&< \left(\frac{p N^{(q-2k)/2}}{1-N^{-1/2}} + N^{(q+1-2k)/2}  \left\lfloor \frac{s+q+2-4k}{q+1-2k} \right\rfloor \right) \frac{s^{2k-1}}{(2k-1)!} \frac{\left(\sqrt{N}+v\right)^{q+1-2k} }{(q+1-2k)!}\notag \\
&<  \left(\frac{pN^{-1/2}}{1-N^{-1/2}} +\left\lfloor \frac{s+q+2-4k}{q+1-2k} \right\rfloor \right) \frac{s^{2k-1}}{(2k-1)!} \frac{\left(N+v\sqrt{N}\right)^{q+1-2k}}{(q+1-2k)!}. \label{eq:lastsum4}
\end{align}

The combinatorial bounds above, by construction, did not depend on the initial graph $G$.  Therefore, combining (\ref{eq:lastsum1}), (\ref{eq:lastsum2}), (\ref{eq:lastsum3}) and (\ref{eq:lastsum4}),  and employing (\ref{eq:sigmadef}), (\ref{eq:collatzwieland}) and (\ref{eq:supGsum}), along with \begin{equation}
v < qp + s
\end{equation}
and other simple inequalities, we obtain 
\begin{align}
\nu_{s^\prime s}^{(p)} &< \nu_1 + \nu_2
\end{align}
where \begin{subequations}\begin{align}
\nu_1 &= \left(\frac{N+(qp+s)\sqrt{N}}{N-q}\right)^{q+1-2k} \left(\frac{q}{N-q}\right)^{2k-2} \frac{2s^{2k-1}}{q(2k-1)!} \left(\frac{pN^{-1/2}}{1-N^{-1/2}} + \frac{s+q+2-4k}{q+1-2k} \right) \\
\nu_2 &= \frac{(q-1)!}{(N-q)^{(q-2)/2}}\left(\frac{\sqrt{N}+(p+1)q+s}{\sqrt{N-q}}\right)^q  \left(\frac{p}{1-N^{-1/2}} + \frac{(s+q)^{q+1-2k}}{(2k-1)!(q+1-2k)!}\right)^2.
\end{align}\end{subequations}
Now we simplify.  Using that $N-q > \frac{1}{2}N$ from (\ref{eq:qlessN2}), along with $p<1+\kappa s \log N$,
\begin{align}
\nu_1 &< \left(1 + 2\frac{q + s\sqrt{N}(1+q\kappa \log N) + q\sqrt{N}}{N}\right)^{q+1-2k} \left(\frac{2q}{N}\right)^{2k-2} \frac{2s^{2k-1}}{q(2k-1)!} \left(\frac{s+q-2}{q+1-2k} + \frac{2(p+1)}{\sqrt{N}} \right) \notag \\
&< \frac{s+q-2}{q+1-2k} \frac{2s}{q} \left(\frac{2qs}{N}\right)^{2k-2} \exp\left[ \frac{5q^2 \kappa s \log N}{\sqrt{N}}\right].
\end{align}
In the second line, we used (\ref{eq:annoying}), $(s+q-2)\ge q+1-2k$, and $1+x < \mathrm{e}^x$, to simplify further.   Next,  \begin{align}
\nu_2 &< \frac{2^{(q-2)/2}(q-1)!}{N^{(q-2)/2}} \left(\sqrt{2} \left(1 + \frac{s+2q+q\kappa s \log N}{\sqrt{N}} \right)\right)^q\left(\frac{2\kappa s \log N}{\sqrt{N}} + \frac{(s+q)^{q+1-2k}}{(2k-1)!(q+1-2k)!} \right) \notag \\
&< \frac{2^q (q-1)! (s+q)^{q+1-2k}}{N^{(q-2)/2}}\exp\left[ \frac{2q^2 \kappa s \log N}{\sqrt{N}}\right] 
\end{align}
where in the second line we used the fact that $2\kappa s\log N + 2 < \sqrt{N}(s+q)^{q+1-2k}$, since $q+1-2k>0$ and we assumed (\ref{eq:kappasqrtNN}).   Making a few final simplifications, we obtain (\ref{eq:lma13}).
\end{proof}

The last step to prove our theorem is to simply bound $\mathcal{K}_l$.  Recall the relation between $l$ and $s$ defined in (\ref{eq:ldef}).  With probability no smaller than $\mathbb{P}_{\mathrm{success}}$ we may invoke Lemma \ref{lma13} as a bound on every $\mu_{s^\prime s}$.  Hence from Proposition~\ref{propmuss} and (\ref{eq:Kldef}), \begin{align}
\mathcal{K}_l &\le \max \left \lbrace \max_{s\in R_l} \sum_{s^\prime \in R_{l+1}} \sqrt{\mathrm{e}^{\theta + 2/\kappa}\nu_{s^\prime s}^{(p)}},  \max_{s^\prime \in R_{l+1}} \sum_{s\in R_{l}}  \sqrt{\mathrm{e}^{\theta + 2/\kappa}\nu_{s^\prime s}^{(p)}}  \right\rbrace \notag \\
&< \mathrm{e}^{\theta/2+1/\kappa} \sum_{k=1}^{q/2-1} \sqrt{\nu_{1+(q-2)l + q+2-4k, 1+(q-2)l}^{(p)}} \notag \\
&<\mathrm{e}^{\theta/2+1/\kappa} \left[ \sqrt{\frac{s+q-2}{q-1} \frac{2s}{q}} \dfrac{1}{\displaystyle 1 -\frac{2q^2s}{N}}   + q\sqrt{\frac{2^q (q-1)! (s+q)^{q-1}}{N^{(q-2)/2}}} \right] \exp\left[ \frac{5q^2 \kappa s \log N}{2\sqrt{N}}\right] \notag \\
&< \mathrm{e}^{\theta/2+1/\kappa} \left[ \sqrt{\frac{2(1+(q-2)(l+1))(1+(q-2)l)}{q(q-1)}} \dfrac{1}{\displaystyle 1 -\frac{2q^3(l+1)}{N}}   + \frac{2^{q} q! (l+1)^{(q-1)/2}}{N^{(q-2)/4}} \right] \notag \\ 
&\;\;\;\;\;\; \times \exp\left[ \frac{5q^3 \kappa (l+1) \log N}{2\sqrt{N}}\right] \notag \\
&<c (l+1)
\end{align}
where in the second line, we have used that the largest of the bounded transition rates come from the maximal value of $s$; in the third line, we used concavity of $\sqrt{x}$; in the fourth line, we plugged in for $s=1+(q-2)l$ and used further elementary inequalities;  in the fifth line, we used $l+1\le M$ along with various other elementary inequalities to further weaken the bound, and the values of $c$ and $M$ from (\ref{eq:theorcm}).
\end{proof}

Combining Corollary \ref{corolscramble} and Theorem \ref{theorSYK}, we immediately obtain 

\begin{cor}
If there exist $\gamma,\delta  \in \mathbb{R}^+$ such that \begin{equation}
q < \gamma N^{(1-\delta)/6},
\end{equation} 
then as $N\rightarrow \infty$ in the SYK ensemble, there is a constant $b\in \mathbb{R}^+$ such that the operator scrambling time obeys the fast scrambling conjecture: \begin{equation}
t_{\mathrm{s}} \ge b\log N.
\end{equation}
\end{cor}

This is the first complete proof of a logarithmic lower bound on the operator scrambling time in a Hamiltonian quantum system whose many-body Lyapunov exponent was expected to be finite.  But we do not expect that these constant prefactors in these bounds are tight.

One shortcoming of our proof, which perhaps has physical consequences, is that we were only able to demonstrate fast scrambling and the diverging growth time of operators if $q \propto N^{(1-\delta)/6}$ for some $\delta \in \mathbb{R}^+$.   Yet it may be the case that fast scrambling is robust so long as $q \propto N^{(1-\delta)/2}$ \cite{cotler}.   However, the SYK model is likely more physical at finite $q$ anyway (e.g. $q=4$), so this may be a relatively minor point.  The large $q$ limit is  a further simplification used in the theoretical physics literature to simplify Feynman diagrammatics.

We make one last remark.  The actual distribution of eigenvalues of $M_{s^\prime s}$ (for any $s>1$) is not sharply peaked around the mean value which we have overestimated.  Instead, the distribution of eigenvalues is highly peculiar, with only a small fraction of eigenvalues, which we conjecture is $\mathrm{O}(N^{1-s})$ for $s<(q-2)M$, within an O(1) factor of $\mathcal{K}_{s^\prime s}$.   We conjecture that the maximal eigenvector of $M_{s^\prime s}$ is dominated by treelike factor graphs, analogous to (\ref{eq:graphtree}), with $\mathrm{O}(s/q)$ leaves attached to a root which connects to a single fermion.   These are precisely the graphs associated to a growing operator which started from a single fermion.  Indeed, explicit calculations confirm that such treelike graphs have significantly larger weight in $\mathrm{tr}(M_{s^\prime s}^p)$ for $p=\kappa s \log N$.  It appears as though the fastest growing operators of average size $\bar s$ is a single fermion operator $\psi_1(t)$, evolved to an appropriate time $t$. It would be interesting if this set of conjectures can be proven or disproven.

\subsection{Comparison with perturbation theory}\label{sec:pert}
Let us now compare our bounds to prior calculations in the SYK model using perturbation theory.   First, let us discuss the Lyapunov exponent as $N\rightarrow \infty$.  We have found that \begin{equation}
\lambda \le 2\sqrt{\frac{2(q-2)}{q}},
\end{equation}
a slight improvement over \cite{chen1}.   It is known analytically that \cite{stanford1802} \begin{equation}
\lambda_{\mathrm{perturbative}} = 2 \;\;\; (q=\infty),
\end{equation}
implying that our result has over estimated the true value by a factor of $\sqrt{2}$.  

\cite{stanford1802} also argued that the block probabilities $P_l(t)$ took the form \begin{subequations}\label{eq:SYKopprob}\begin{align}
P_0(t) &\approx 1 - \frac{4}{q} \log \cosh t + \cdots , \\ 
P_l(t) &\approx \frac{2}{lq} (\tanh t)^{2l} + \cdots.\;\;\;\; (l>0)
\end{align}\end{subequations}
at leading order in a large $N$ and large $q$ expansion (with no bound on the subleading corrections, denoted above as $\cdots$).  It is interesting to compare this with the following result: \begin{prop}
Consider the quantum walk of $|\varphi(t)\rangle$ generated by (\ref{eq:Haux}) with \begin{equation}
K_l(t) = c(l+1),  \label{eq:Klcl2}
\end{equation}
on the half-line where $N^\prime = \infty$.   Then \begin{equation}
P_l(t) = (\tanh (ct))^{2l} \mathrm{sech}^2 (ct).  \label{eq:optqw}
\end{equation} \label{propquantumwalk}
\end{prop}
\begin{proof}
Without loss of generality, we rescale time so that $c=1$.  Then, we repackage (\ref{eq:prop2dos}) using generating functions: \begin{equation}
G(z,t) :=  \sum_{l=0}^\infty z^{l+1} P_l(t),
\end{equation}
so that (\ref{eq:prop2dos}) with (\ref{eq:Klcl2}) implies that \begin{equation}
\frac{\partial G}{\partial t} = z^2 \frac{\partial G}{\partial z} - z \frac{\partial}{\partial z} \left(\frac{G}{z}\right) = \left(z^2-1\right)\frac{\partial G}{\partial z} + \frac{G}{z}.
\end{equation}
This equation is solved by the method of characteristics.  The characteristic curves $z(t)$ solve the differential equation \begin{equation}
\frac{\mathrm{d}z}{\mathrm{d}t} = \left(1-z^2\right).
\end{equation}
With initial condition $z(0)=r$, we find \begin{equation}
t = \frac{1}{2} \log \frac{(1-r)(1+z)}{(1+r)(1-z)},
\end{equation}
or \begin{equation}
r = \frac{z\cosh t - \sinh t}{\cosh t - z\sinh t}.
\end{equation}
Solving the equation \begin{equation}
\frac{\partial G(r,t)}{\partial t} = \frac{G}{z}
\end{equation}
with $G(r,0)=r$ (corresponding to $P_0(0)=1$): \begin{equation}
\log \frac{G}{r} = \int\limits_0^t \mathrm{d}t^\prime \frac{\cosh t^\prime  + r\sinh t^\prime}{\sinh t^\prime + r\cosh t^\prime} = \log \frac{\sinh t + r\cosh t }{r}.
\end{equation}
Thus, \begin{equation}
G(z,t)= \frac{z\mathrm{sech}t}{(1 - z\tanh t)},
\end{equation}
which leads to (\ref{eq:optqw}) upon Taylor expanding and employing (\ref{eq:varphisqrt}).
\end{proof}

Some of the discrepancy between (\ref{eq:optqw}) and (\ref{eq:SYKopprob}) can be accounted for by our sloppy overestimate of $K_l(t)$ in the SYK model.  In particular, a more careful analysis demonstrates that $K_0(t) \lesssim \sqrt{2/q}$ and $K_l(t) \lesssim l$.  However, this slow first step does not change our estimate for the Lyapunov exponent.

This result is highly suggestive that the qualitative structure of the growing operator distribution in the SYK model, calculated perturbatively, is not substantially modified by non-perturbative physics.  Rather it appears quite similar to an ``optimal" quantum walk that locally maximizes the transition rates from one operator size to the next.\footnote{However, such an ``optimal" quantum walk likely does not maximize the probability of large size operators, and perhaps does not even optimize the time-dependent average size.}  This may imply some universality to the patterns of operator growth in random regular $q$-local quantum systems.   If such universality exists, it may  have interesting implications for quantum gravity.

\section{Conclusion}
In this paper, we have proven the fast scrambling conjecture in the SYK model with a finite but large number $N$ of degrees of freedom.   While this result is not physically surprising, it is pleasing to have a mathematically careful derivation of this result.  We also expect that the methods developed here will lead to further advances in our technology for bounding quantum information dynamics and operator growth \cite{chen1, chen2} beyond the Lieb-Robinson theorem \cite{liebrobinson, hastings}.

We would like to say that  our demonstration of the robustness of operator growth  to  non-perturbative physics in at least one holographic model is a signature that the bulk geometry is semiclassical and that non-perturbative  fluctuations in quantum gravity are provably mild.  Unfortunately, this remains a conjecture, as the emergent geometry arises at finite temperature.  It would be interesting if our methods can be generalized to finite temperature states. 

Lastly, we expect these techniques are useful for designing and constraining toy models of quantum gravity which can be experimentally studied using quantum simulation \cite{Garttner2017, Li2017}.  At the very least, any tentative model must reproduce the exponential growth in operator size which is a  hallmark of particles falling towards black hole horizons.  Our  methods will not only bound the Lyapunov exponent of any proposed model, but also check whether the full time  evolution of  the operator size distribution could be consistent with a theory of quantum gravity.

\section*{Acknowledgments}
This work was supported by the University of Colorado.

\bibliographystyle{unsrt}
\bibliography{sykbib}

\end{document}